\documentclass[12pt]{article}
\usepackage{latexsym,amsmath,amsfonts,natbib,bm}
\usepackage{amssymb,amsthm,rotating,threeparttable}
\usepackage{graphicx,subfigure,algorithmicx}
\usepackage{algpseudocode}
\usepackage{amsmath}
\usepackage{setspace}
\usepackage{lscape}
\usepackage[displaymath, mathlines]{lineno}
\usepackage{appendix}
\usepackage{comment,verbatim}
\usepackage{wasysym }
\usepackage{epstopdf,amsxtra}

\usepackage{epstopdf}
\usepackage{lineno}
\usepackage{refcount} 
\usepackage{float}
\usepackage{stackengine}
\usepackage{url}

\def\spacingset#1{\renewcommand{\baselinestretch}%
{#1}\small\normalsize} \spacingset{1}

\include{def}
\def\dispmuskip{\thinmuskip= 3mu plus 0mu minus 2mu \medmuskip=  4mu plus 2mu minus 2mu \thickmuskip=5mu plus 5mu minus 2mu}
\def\textmuskip{\thinmuskip= 0mu                    \medmuskip=  1mu plus 1mu minus 1mu \thickmuskip=2mu plus 3mu minus 1mu}
\def\beq{\dispmuskip\begin{equation}}    \def\eeq{\end{equation}\textmuskip}
\def\beqn{\dispmuskip\begin{displaymath}}\def\eeqn{\end{displaymath}\textmuskip}
\def\bea{\dispmuskip\begin{eqnarray}}    \def\eea{\end{eqnarray}\textmuskip}
\def\bean{\dispmuskip\begin{eqnarray*}}  \def\eean{\end{eqnarray*}\textmuskip}

\def\paradot#1{\vspace{1.3ex plus 0.7ex minus 0.5ex}\noindent{\bf\boldmath{#1.}}}

\newtheorem{assumption}{Assumption}
\newtheorem{corollary}{Corollary}
\newtheorem{lemma}{Lemma}

\newtheorem{remark}{Remark}
\newcommand*\samethanks[1][\value{footnote}]{\footnotemark[#1]}

\newcommand{\MH}{Metropolis-Hastings}

\newcommand{\eps}{\epsilon}
\newcommand{\wh}{\widehat}
\newcommand{\wt}{\widetilde}

\newcommand{\ov}{\overline}

\def\E{\mathbb{E}}                         
\def\Var{{\mathbb{V}}}
\def\G{{\mathcal{G}}}
\def\Cov{{\rm Cov}}
\def\Corr{{\rm Corr}}
\def\P{{\rm P}}                         

\def\G{\Gamma}
\def\a{\alpha}
\def\T{\Theta}

\def\s{\sigma}

\def\t{\theta}
\def\b{\beta}
\def\l{\lambda}

\def\CT{\text{\rm CT}}

\def\IF{\text{\rm IF}}

\def\opt{\text{\rm opt}}
\newcommand{\half}{\frac12}
\def\N{\mathcal N}

\def\G{\cal G}

\def\d{{\rm d}\,}
\def\sgn{\text{\rm sgn}}

\newcommand{\usub}[1]{{{\bm u}}_{(#1)}}
\newcommand{\Usub}[1]{{{\bm U}}_{(#1)}}
\newcommand{\zsub}[1]{{{ z}}_{(#1)}}
\newcommand{\zvec}{\mathcal{Z}}
\newcommand{\gsub}[1]{{{ g}}_{(#1)}}
\newcommand*\patchAmsMathEnvironmentForLineno[1]{%
  \expandafter\let\csname old#1\expandafter\endcsname\csname #1\endcsname
  \expandafter\let\csname oldend#1\expandafter\endcsname\csname end#1\endcsname
  \renewenvironment{#1}%
     {\linenomath\csname old#1\endcsname}%
     {\csname oldend#1\endcsname\endlinenomath}}%
\newcommand*\patchBothAmsMathEnvironmentsForLineno[1]{%
  \patchAmsMathEnvironmentForLineno{#1}%
  \patchAmsMathEnvironmentForLineno{#1*}}%
\AtBeginDocument{%
\patchBothAmsMathEnvironmentsForLineno{equation}%
\patchBothAmsMathEnvironmentsForLineno{align}%
\patchBothAmsMathEnvironmentsForLineno{flalign}%
\patchBothAmsMathEnvironmentsForLineno{alignat}%
\patchBothAmsMathEnvironmentsForLineno{gather}%
\patchBothAmsMathEnvironmentsForLineno{multline}%
}

\makeatletter
\def\@seccntformat#1{\@ifundefined{#1@cntformat}%
   {\csname the#1\endcsname\quad}  
   {\csname #1@cntformat\endcsname}
}
\let\oldappendix\appendix 
\renewcommand\appendix{%
    \oldappendix
    \newcommand{\section@cntformat}{\appendixname~\thesection\quad}
}
\makeatother
\usepackage{cleveref} 
\usepackage{xr}
\begin{document}
\title{The Block Pseudo-Marginal Sampler}
\author{M.-N. Tran\thanks{Discipline of Business Analytics, University of Sydney}
\and R. Kohn \thanks{School of Economics, UNSW School of Business}
\and M. Quiroz \samethanks
\and M. Villani\thanks{Department of Computer and Information Science, Link\"{o}ping University}
}
\maketitle


\begin{abstract}
The pseudo-marginal (PM) approach
is increasingly used for Bayesian inference in statistical models, where the likelihood
is intractable but can be estimated unbiasedly.
\cite{deligiannidis:doucet:pitt:2015} show how the PM approach can be made much more
efficient by correlating the underlying Monte Carlo (MC) random numbers used to form the estimate of the likelihood at the current and
proposed values of the unknown parameters. Their approach greatly speeds up the standard PM algorithm,
as it requires a much smaller number of samples or particles to form the optimal likelihood estimate.
Our paper presents an alternative implementation of the correlated PM approach, called the
block PM, which divides the underlying random numbers into blocks so that the likelihood estimates for the
proposed and current values of the parameters only differ by the random numbers in one block.
We show that this implementation of the correlated PM can be much more efficient for some specific problems
than the implementation in \cite{deligiannidis:doucet:pitt:2015}; for example
when the likelihood is estimated by subsampling or the likelihood is a product of terms each of which is given by an integral which
can be estimated unbiasedly by randomised quasi-Monte Carlo. Our article provides methodology and guidelines for efficiently implementing the
block PM. A second advantage of the the block PM is that it provides a direct way to control the correlation between the logarithms of the estimates of the
likelihood at the current and proposed values of the parameters than the implementation in \cite{deligiannidis:doucet:pitt:2015}.
We obtain methods and guidelines for selecting the optimal number of samples based on idealized but realistic assumptions.

\paradot{Keywords}
Intractable likelihood; Unbiasedness; Panel-data; Data subsampling; Randomised quasi-Monte Carlo.
\end{abstract}

\spacingset{1.45} 
\section{Introduction} \label{Sec:introduction}
In many statistical applications the likelihood is analytically or computationally intractable,
making it difficult to carry out Bayesian inference. An example of models where the likelihood is often intractable are generalised linear mixed models (GLMM) for longitudinal data, where random effects are used to account for the dependence between the observations measured on the same individual \citep{Fitzmaurice:2011,Bartolucci:2012}. The likelihood is intractable because it is an integral over the random effects, but it can be easily estimated unbiasedly using importance sampling. The second example that uses a variant of the unbiasedness idea, is that of
unbiasedly estimating the log-likelihood by subsampling, as in \cite{quiroz:villani:kohn::2016}.
Subsampling is useful when the log-likelihood is a sum of terms, with each term expensive to evaluate, or when there is a very large number of such terms.
\cite{quiroz:villani:kohn::2016}
estimate the log-likelihood unbiasedly in this way and then bias correct the resulting likelihood estimator to use within a PM algorithm. See also \cite{Quiroz:2016} for an alternative subsampling approach using the Poisson estimator to obtain an unbiased estimator of the likelihood
and \cite{quiroz:DelayedAcc} for subsampling with delayed acceptance.
State space models are a third class of models
where the likelihood is often intractable but can be unbiasedly estimated using
an importance sampling estimator \citep{Shephard:1997,Durbin:1997} or a particle filter estimator \citep{DelMoral:2004, Andrieu:2010}.

It is now well known in the literature that a direct way to overcome the problem of working with an
intractable likelihood is to estimate the likelihood unbiasedly and use this estimate within a Markov chain Monte Carlo (MCMC)
simulation on an expanded space that includes the random numbers used to construct the likelihood estimator.
This was first considered by \cite{LinLiuSloan:2000} in the Physics literature and \cite{Beaumont:2003} in the Statistics literature.
It was formally studied in \cite{Andrieu:2009}, who called it the pseudo-marginal (PM) method and gave conditions
for the chain to converge. \cite{Andrieu:2010} use the PM approach for inference in state space models
where the likelihood is estimated unbiasedly by the particle filter.
\cite{Flury:2011} give an excellent discussion with illustrative examples of PM.
\cite{Pitt:2012} and \cite{Doucet:Pitt:Deligiannidis:Kohn} analyse the effect of estimating the likelihood
and show that the variance of the log-likelihood estimator should be around 1
to obtain an optimal tradeoff between the efficiency of the Markov chain
and the computational cost. See also \cite{Sherlock:2015}, who consider random walk proposals for the
parameters, and show that the optimal variance of the log of the likelihood estimator can be somewhat
higher in this case.

A key issue in estimating models by standard PM is that the variance of the log of the estimated likelihood grows linearly with the
number of observations $T$. Hence, to keep the variance of the log of the estimated likelihood small and around 1 it is necessary
for  the number of samples $N$, used in constructing the likelihood estimator, to increase in
proportion to $T$, which means that PM requires $O(T^2)$ operations at every MCMC iteration. Starting with \cite{Lee:2010}, several authors have noted that PM methods can benefit from updates of the underlying random numbers used to construct the estimator that correlate the numerator and denominator of the PM acceptance ratio \citep{deligiannidis:doucet:pitt:2015, Dahlin:2015}. \cite{Lee:2010} propose to use MH moves that alternate between i) updating the parameters conditional on the random numbers and ii) updating the random numbers conditional on the parameters. The effect is that the random numbers are fixed at some iterations hence inducing a high correlation when the parameters are updated.
However, this approach gives no correlation whenever the random numbers are updated as they are all updated simultaneously. Unless the variance of the likelihood estimator is very small, the \cite{Lee:2010} PM sampler is likely to quickly get stuck.
The \cite{Lee:2010} proposal is a special case of \cite{Stramer:2011}, which we discuss in more detail in Section~\ref{sub:Diffusion process example}.

 \cite{deligiannidis:doucet:pitt:2015} propose a better way to induce correlation between
the numerator and denominator of the MH ratio by
correlating the Monte Carlo (MC) random numbers used in constructing the estimators of the likelihood at the
current and proposed values of the parameters.
We call this approach the correlated PM (CPM) method,
and we call the standard PM the independent PM (IPM) method,
as  a new independent set of MC random numbers is used each time the likelihood is estimated.
\cite{deligiannidis:doucet:pitt:2015} show that by inducing a high correlation between these ensembles of MC random numbers
it is only necessary to increase the number of samples $N$ in proportion to $T^\half$, reducing
 the CPM algorithm to $O(T^{3/2})$ operations per iteration. This is likely to be an important breakthrough in the ability
of PM to be competitive with more traditional MCMC methods.
\cite{Dahlin:2015}  also propose a CPM algorithm but did not derive any optimality results.

Our paper proposes an alternative implementation of the CPM approach, called the block pseudo-marginal (BPM), that can be much more efficient than CPM for some specific problems.
The BPM approach divides the set of underlying random numbers into blocks and updates the unknown parameters {\em jointly} with one of these blocks at any one iteration
 which induces a positive correlation between the numerator and denominator of the MH acceptance ratio, similarly to the CPM. This correlation reduces the variation in the \MH{} acceptance probability, which helps the underlying Markov chain of iterates to mix well even if highly variable estimates of the likelihood are used.
This means that a much smaller number of samples is needed than if all the underlying random variables are updated independently each time.
We derive methodology and  guidelines for selecting an optimal number of samples in BPM based on idealized but plausible assumptions.

Although CPM is a more general approach than BPM, we believe that the BPM approach
has the following advantages over the CPM method in specific settings.

\noindent
{\sf (i)~Efficient data handling}. For some applications such as data subsampling \citep{Quiroz:2016, quiroz:villani:kohn::2016} the BPM method
can take less CPU time than the IPM and CPM as it is unnecessary to work with the whole data set, and  it is also unnecessary
to generate the full set of underlying random numbers in each iteration.

\noindent
{\sf (ii)~Randomised quasi Monte Carlo}. The BPM method offers a natural way to estimate integrals
unbiasedly using randomized quasi Monte Carlo (RQMC) sampling instead of Monte Carlo (MC).
In many cases, numerical integration using RQMC achieves a better convergence rate than MC.
Using RQMC has recently proven successful in the intractable likelihood literature; see, e.g., \cite{Gerber:2015} and \cite{Gunawan:2016}.
We show that, if RQMC is used to estimate the likelihood, the optimal number of samples required at each iteration of BPM is approximately $O(T^{7/6})$,
compared to $O(T^{3/2})$ in the CPM approach of \cite{deligiannidis:doucet:pitt:2015} who use MC.
Correlating randomised quasi numbers in CPM is challenging, as it is difficult to preserve the desirable uniformity
properties of RQMC. See \cite{Gunawan:2016} for a first attempt at correlating quasi random numbers in CPM.

\noindent
{\sf (iii)~Preservation of correlation}.  If the likelihood can be factorised into blocks,
then the correlation of the logs of the estimated likelihoods at the current and proposed values
is close to $1-1/G$, where $G$ is the number of blocks in the blocking approach. That is, the correlation between the proposed and current values of the log likelihood
estimates is controlled directly rather than indirectly and nonlinearly through the correlated ensembles of random numbers.
This property of correlation preservation is a potentially important issue as the log of the estimated likelihood can be a very nonlinear transformation of the underlying random variables, and hence correlation may not be preserved in CPM.

As we note above, CPM is a more general approach than BPM because it can be used in applications where blocking cannot be applied such
as correlating the number of terms used in the Poisson estimator when debiasing \citep{Quiroz:2016}. Second, if the likelihood cannot be factored into a number of
independent blocks such as in nonlinear state space models, then it is unclear whether BPM has any advantages over CPM. Finally,
in some problems such exact subsampling, it will be useful to combine BPM and CPM
to obtain a more efficient correlated PM approach \citep{Quiroz:2016}.

The paper is organized as follows.
Section \ref{Sec: BPM} introduces the BPM approach
and Section \ref{SS: analysis of block} presents methodology and guidelines for efficiently implementing the block PM.
Section \ref{sec:applications} presents applications.
Section \ref{Sec:Conclusion} concludes.
There is an an online supplement to the paper containing five appendices. Appendix~\ref{Proofs} gives proofs  of all the results in the paper.
Appendix~\ref{sec: large sample} gives some large-sample properties of the BPM for panel data.
Appendix~\ref{app: CT derivation} derives the expression for computing time.
Appendix~\ref{S: A toy example} presents an illustrative toy example. Appendix~\ref{S: further applications}
gives two further applications.

\section{The block pseudo-marginal approach}\label{Sec: BPM}
\subsection{The independent PM approach} \label{SS: the IBM}
Let $y$ be a set of observations with density $ L(\theta):=p(y|\theta)$,
where $\theta\in\Theta$ is the vector of unknown parameters and
let $p_\Theta(\theta)$ be the prior for $\theta$.
 We are interested in sampling from the posterior $\pi(\theta)\propto p_\T(\theta)L(\theta)$ in
models where the likelihood $L(\theta)$ is analytically or computationally intractable.
Suppose that  $L(\theta)$ can be estimated by a nonnegative and unbiased estimator $\wh L(\theta, \bm u)$,
which we sometimes write as $\wh L(\theta)$,
with $ \bm u \in \mathbb{U}$ the set of
random numbers used to compute $\wh L(\theta)$.
The likelihood estimator $\wh L(\theta,\bm  u)$ typically depends on an algorithmic number $N$
that controls the accuracy of $\wh L(\theta, \bm u )$, and is proportional to the cardinality or dimension of the set $\bm u $.
For example, $N$ can be the number of importance samples if the likelihood is estimated by importance sampling,
or $N$ is the number of particles if the likelihood in state space models is estimated by particle filters.
However, for simplicity, we will call $N$ the number of samples throughout.
Denote the density function of $\bm  u$ by $p_U(\cdot)$ and define a joint target  density of $\theta$ and $\bm  u$ as
\beq\label{eq:target 1}
\ov\pi(\theta, \bm u ):=p_\Theta(\theta)\wh L(\theta,\bm u)p_U(\bm u)/\ov L,
\eeq
where $\ov L: = p(y)= \int p(y|\theta) p_\Theta(\theta) \d \theta$ is the marginal likelihood.
$\ov\pi(\theta, \bm u)$ admits $\pi(\theta)$ as its marginal density because  $\int\wh L(\theta, \bm u)p_U(\bm u)\d {\bm u}=L(\theta)$ by the unbiasedness
of $\wh L(\theta, \bm u)$.
Therefore, we can obtain samples from the posterior $\pi(\theta)$  by sampling from $\ov\pi(\theta, \bm u)$.

Let $q_\Theta(\theta|\theta')$ be a proposal density for $\theta$, conditional on the current state $\theta'$.
Let $\bm u'$ be the corresponding current set of random numbers used to compute $\wh L(\theta', \bm u')$.
The independent PM algorithm generates samples from $\pi(\theta)$ by generating a Markov chain with invariant density $\ov\pi(\theta,\bm u)$
using the \MH{} algorithm with proposal density $q(\theta, \bm u|\theta',  \bm u')=q_\Theta(\theta|\theta')p_U(\bm u)$.
The proposal $(\theta,\bm u)$ is accepted with probability
\begin{align}\label{eq:acceptance}
\alpha(\theta',\bm u'; \theta, \bm u)& :=\min\left(1,\frac{\ov\pi(\theta, \bm u)}{\ov\pi(\theta', \bm u')}\frac{q(\theta', \bm u'|\theta, \bm u)}{q(\theta, \bm u|\theta', \bm u')}\right)
 =\min\left(1,\frac{p_\Theta(\theta)\wh L(\theta,\bm u)}{p_\T(\theta')\wh L(\theta', \bm u')}\frac{q_\Theta(\theta'|\theta)}{q_\Theta(\theta|\theta')}\right),
\end{align}
which is computable.
In the IPM scheme, a new independent set of MC random numbers $\bm u$ is generated
each time the likelihood estimate is computed,
and it is usually unnecessary to store $\bm  u$ and $\bm u'$.

\cite{Pitt:2012} and \cite{Doucet:Pitt:Deligiannidis:Kohn}
show for the IPM algorithm that the variance of $\log\;\wh L(\theta,\bm u)$
should be around 1 in order to obtain an optimal tradeoff between the
computational cost and efficiency of the Markov chain in $\theta$ and $\bm u$.
However, in some problems it may be prohibitively expensive to take a $N$ large enough to ensure
that $\Var(\log\;\wh L(\theta,\bm u))\approx 1$.

\subsection{The block PM approach}\label{subsec:BPM}
In the block PM algorithm, instead of generating a new set $\bm u$
when estimating the likelihood as in the independent PM,
we update $\bm u$ in blocks.
Suppose we divide the set of variables $\bm u$ into $G$ blocks $\usub{1},...,\usub{G}$, with $\usub{j} \in \mathbb{U}_j$, $j=1, \dots, G$, and $\mathbb{U}:= \mathbb{U}_1 \times \mathbb{U}_2 \times \cdots \times \mathbb{U}_G$. We construct $p_U(\bm u) := \prod_{j=1}^G p_{\Usub{j}}(\usub{j}) $. We rewrite
the extended target \eqref{eq:target 1} as
\beq\label{eq:target 2}
\ov\pi(\theta,\usub{1:G})=p_\Theta(\theta)\wh L(\theta,\usub{1:G})\prod_{j=1}^G p_{\Usub{j}}(\usub{j}) /\ov L,
\eeq
and propose to update $\theta$ and just one block of the $\usub{j},~\,  j=1, \dots, G$.
Let ${\bm u}^\prime :=( \usub{1}^\prime, \dots, \usub{G}^\prime )$ be the current value of $\bm u$. Then the
proposal distribution for $\bm u$ is
\begin{align}\label{eq: proposal q for u}
q(\d \usub{1:G}|\usub{1:G}^\prime)&: =
\sum_{i=1}^G \omega_i p_{\Usub{i}} (\usub{i}) \d \usub{i} \prod_{j \neq i} \delta_{\usub{j}^\prime} (\d \usub{j}),
\end{align}
with $\omega_i=1/G$ for all $i$ and $\delta_a(\d
\bm b)$ is the delta measure concentrated at $\bm a$. The next lemma expresses the acceptance probability
\eqref{eq:acceptance} of the PM scheme with proposal density \eqref{eq: proposal q for u}.

\begin{lemma} \label{lemma: accept prob}
The acceptance probability  \eqref{eq:acceptance} of the PM scheme with proposal distribution \eqref{eq: proposal q for u} is
\beq\label{eq:acceptance 2}
\min\left(1,\frac{p_\Theta(\theta)\wh L(\theta,\usub{1:k-1}^\prime,\usub{k},\usub{k+1:G}^\prime)}{p_\T(\theta')\wh  L(\theta',\usub{1:G}^\prime)}\frac{q_\Theta(\theta'|\theta)}{q_\Theta(\theta|\theta')}\right),
\eeq
and is computable.
\end{lemma}

This allows us to carry out MCMC, similarly to other component-wise MCMC schemes;
see, e.g., \cite{Johnson:2013}.
We show in the proof of part (ii) of Lemma~\ref{lem:acceptance prob} that by fixing all the $\usub{j}$ except $\usub{k}$, the variance of the log of the ratio
of the likelihood estimates is reduced. This reduction in variance may help the chain mix well, although there is a potential tradeoff
between block size and mixing as the $\usub{k}$ mix more slowly. Lemma~\ref{lemma:uncorrelated in n} shows that for large sample sizes,
moving the $\usub{k}$ slowly does not impact the mixing of the $\theta$ iterates because $z(\theta, \bm u )$ and $\theta$ are uncorrelated. Furthermore,
we have also found this to be the case empirically for moderate and large sample sizes. These comments of slower mixing also apply to the
correlated PM sampler.

\subsection{Randomized quasi Monte Carlo} \label{SS: rqmc}
RQMC has recently received increasing attention in the intractable likelihood literature \citep{Gerber:2015,Tran:2015,Gunawan:2016}. See
\cite{Niederreiter1992} and \cite{Dick2010} for a thorough treatment.
Typically, MC methods estimate a $d$-dimensional integral of interest based on i.i.d. samples from the uniform distribution $\mathcal{U}(0,1)$.
RQMC methods are  alternatives that choose {\em deterministic} points in $[0,1)$ evenly in the sense that they minimize the so-called star-discrepancy of the point set.
Randomized MC then injects randomness into these points such that the resulting points preserve the low-discrepancy property and, at the same time, they marginally have a uniform distribution.
\cite{Owen:1997} shows that the variance of RQMC estimators is of order $N^{-3}(\log N)^{d-1}=O(N^{-3+\epsilon})$ (where $d$ is the dimension of the argument in the integrand)
for any arbitrarily small $\epsilon>0$, compared to $O(N^{-1})$ for plain MC estimators, with $N$ the number of samples.
Central limit theorems for RQMC estimators are obtained in \cite{Loh2003}.

In block PM with RQMC numbers, the set $\bm u$ will be RQMC numbers instead of MC numbers.
In this paper, RQMC numbers are generated using the scrambled net method of \cite{Matousek:1998}.

\subsection{The correlated PM}\label{SS: CPM}
Instead of updating $\bm u$ in blocks, \cite{deligiannidis:doucet:pitt:2015} move $\bm u$ slowly
by correlating the proposed $\bm u$ with its  current value $\bm u^\prime$.
Suppose that the underlying MC numbers $\bm u$ are standard univariate normal variables
and $\varrho>0$ is a number close to 1.  \cite{deligiannidis:doucet:pitt:2015} set $\bm u=\varrho \bm u'+\sqrt{1-\varrho^2}\bm \eps$
with $\bm \eps$ a vector of standard normal variables of the same size as $\bm u'$.
We note that it is challenging to extend this
correlated PM approach to the case where $\bm u$ are RQMC numbers,
because in the RQMC framework we work with uniform random numbers so that inducing correlation in such numbers may break down their desired uniformity
\citep{Gunawan:2016}.
In contrast,  it is straightforward to use RQMC in the standard way in the block PM.

\section{Properties of the block PM} \label{SS: analysis of block}
Suppose that the likelihood can be written as a product of $G$ independent terms,
\begin{align}\label{eq:panel data llh}
L(\theta) = \prod_{k=1}^G L_{(k)}(\t)\;\;\text{where}\;\; L_{(k)}(\t)= p(y_{(k)}|\theta).
\end{align}
We show in Section \ref{sec:applications} how to apply the  block PM approach when the likelihood cannot be factorised as in \eqref{eq:panel data llh}.
We assume that the $k^{th}$ likelihood term $L_{(k)}(\theta)$ is estimated unbiasedly by $\wh L_{(k)}(\theta, \usub{k})$,
where the $\usub{k}$ are independent with $ \usub{k} \sim p_{\Usub{k}}( \cdot) $.
Let $N_{(k)}$ be the number of samples used to compute $\wh L_{(k)}(\t, \usub{k})$, with
$N := N_{(1)} + \cdots N_{(G)} $. An  unbiased estimator of the likelihood is
\begin{align*}
\wh L(\theta, \bm u ) := \prod_{k=1}^G\wh L_{(k)}(\theta, \usub{k}),
\end{align*}
where $\bm u = \{\usub{1}, \dots, \usub{G} \}$.

\paradot{Example: panel-data models} Consider a panel-data model with $T$ panels, which we divide into $G$ groups $y_{(1)}$,...,$y_{(G)}$, with approximately $T/G$ panels in each. See Section~\ref{sub:PanelData}.

\paradot{Example: big-data} Consider a big-data set with $T$ independent observations, which we divide into $G$ groups $y_{(1)}$,...,$y_{(G)}$, with $T/G$ observations in each. See Section~\ref{SS: data subsampling}.

\subsection{Block PM based on the errors in the estimated log-likelihood} \label{SS: pm based on z's}
Our analysis of the block PM builds on the framework of \cite{Pitt:2012} who provide an analysis of the IPM based on the error in the log of the estimated likelihood.
For any $\theta \in \Theta$, $\usub{k} \in \mathbb{U}_k$, $k=1, \dots, G$, we define
\[\zsub{k}:=\zsub{k}(\theta,\usub{k}):=\log\;\wh L_{(k)}(\theta, \usub{k})-\log\;L_{(k)}(\theta) \quad \text{and} \quad
z(\theta,\usub{1:G}):= \zsub{1}(\theta,\usub{1})+ \cdots +\zsub{G}(\theta,\usub{G}). \]
More generally, for indices $1 \leq i_1 < i_2 < \cdots < i_k\leq G$, we define
\[z(\theta,\usub{i_1:i_k}):= \zsub{i_1}(\theta, \usub{i_1} ) + \zsub{i_2}(\theta, \usub{i_2} ) + \cdots + \zsub{i_k}(\theta, \usub{i_k} ).\]
If \[\bm u=\usub{1:G} \sim \prod_{k=1}^G p_{U_{(k)}}(\cdot), \] then $\zsub{k}(\theta,\usub{k}) $
is the error in the log of the estimated likelihood of the $k$th block
 and $z(\theta, \bm u)$ is the error in the log of the estimated  likelihood.
We now follow \cite{Pitt:2012} and work with the $\zsub{k}$ and $z$ instead of the $\usub{k}$ and $\bm u$, for two reasons.
First, the $\zsub{k}$ and $z$ are scalar, whereas the $\usub{k}$ and $\bm u$ are likely to be high dimensional vectors;
second, the properties of the pseudo-marginal MCMC depend on $\bm u $ only
through $z$.

We use the notation $w\sim \N(a,b^2)$ to mean that $w$ has a normal distribution with mean $a$ and variance $b^2$, and
denote the density of $w$ as $\N(w; a,b^2)$.
Our guidelines for the block PM are based on Assumptions~\ref{ass: assumption on variances}--\ref{ass: perfect proposal}.
\begin{assumption} \label{ass: assumption on variances}
Suppose $ \usub{1}, \dots,  \usub{G}$ are independent and generated from $p_{\Usub{k}}(\cdot)$ for $k=1,...,G$. We assume that
\begin{enumerate}
\item [(i)] For each block $k$, there is a  $\gamma^2_{(k)}(\theta) >0$, an $N_{(k)}>0$ and a $\varpi>0$ such that
\begin{align*}
\Var ( \zsub{k}(\theta, \usub{k})) & =\frac{\gamma^2_{(k)}(\theta)}{N_{(k)}^{2\varpi}}.
\end{align*}
\item [(ii)] For a given $\sigma^2>0$, let $N_{(k)}$ be a function of $\theta$, $\sigma^2$ and $G$ such that
$\Var(\zsub{k}(\theta,\usub{k}))=\sigma^2/G$, i.e. $N_{(k)}=N_{(k)}(\theta, \s^2,G)=[G\gamma^2_{(k)}(\theta) /\sigma^2]^{1/(2\varpi)}$.
Thus, $\sigma^2 = \Var(z(\theta, \bm u)) $ is the variance of the log of the estimated likelihood.
\item [(iii)]
Both $z(\theta, \usub{1:G})$ and $z(\theta, \usub{1:k-1},\usub{k+1:G})$
are normally distributed for each $k$.
\end{enumerate}
\end{assumption}
It is clear from Lemma~\ref{lemm: clt log pi} that $\varpi=1/2$ if the likelihood is estimated using MC,
and $\varpi=3/2-\epsilon$ for any arbitrarily small $\epsilon>0$ if the likelihood is estimated using RQMC.
We note that $N_{(k)}$ is the total number of samples used for the $k$th group, and will usually be different from $N_k$.
In panel-data models and in the diffusion example in Section \ref{sub:Diffusion process example}, $N_{(k)} = (T/G) N_k$ and in the data subsampling example
$N_{(k)} = T/G$.
For the panel-data and subsampling applications, parts (i) and (ii) of
Assumption~\ref{ass: assumption on variances} can be made to hold by construction
 because it is straightforward to
estimate the variance of $\zsub{k}$ accurately for each $k$ and $\theta$.
Part (iii) will usually hold for $G$ large by the central limit theorem (see  Lemma~\ref{lemm: clt log pi}).

\begin{assumption}\label{assum: assumption on additional term}
Suppose that $\usub{k}\sim p_{\Usub{k}}(\cdot) $ and  $(\usub{1:k-1}^\prime,\usub{k+1:G}^\prime)\sim \ov \pi(\cdot|\theta) $  and that $\usub{k}$ is independent of $ \usub{1:k-1}^\prime$ and $\usub{k+1:G}^\prime$. We assume that
$\zsub{k}(\theta, \usub{k}) + z(\theta, \usub{1:k-1}^\prime,\usub{k+1:G}^\prime)$ is normally distributed for a given $\theta$.
\end{assumption}
\begin{remark} \label{remark: remark on assumption additional term}
Assumption~\ref{assum: assumption on additional term} relies on $G$ being large so that the contribution of $\zsub{k}(\theta, \usub{k}) $ is very small compared to that of $z(\theta, \usub{1:k-1}^\prime,\usub{k+1:G}^\prime)$.
If $N_k$ is large, as it is likely to be when $T$ is large (see Lemma~\ref{cor: large T}),
then $\zsub{k}(\theta, \usub{k}) $ is likely to be normally distributed and
then Assumption \ref{assum: assumption on additional term} will hold.
\end{remark}

\begin{assumption}\label{ass: perfect proposal}
We follow \cite{Pitt:2012} and assume a perfect proposal for  $\theta$, i.e.
$q_\Theta(\theta|\theta')=\pi(\theta)$. This proposal simplifies
the derivation of the guidelines for the optimal number of samples,
\end{assumption}
Assuming a perfect proposal leads to a conservative choice of the optimal $\sigma$, both in theory and practice,
in the sense that the prescribed number of samples
is larger than optimal for a poor proposal.
However, such a conservative approach is desirable because
the optimal prescription for the choice of $\sigma$ would be based on idealized assumptions that are unlikely to hold in practice.

Lemma~\ref{lemm: asymptotics} shows that the correlation between the estimation errors in the current and proposed values of $(\theta, \bm u ) $  is directly controlled by
$\rho = 1-1/G$ when blocking. This should be compared with CPM where the correlation is specified on the underlying random numbers $\bm{u}$, but the final effect on the estimation errors is less transparent.
\begin{lemma} [Joint asymptotic distribution of $z$ and $z^\prime$] \label{lemm: asymptotics} Suppose that Assumptions~\ref{ass: assumption on variances} and \ref{assum: assumption on additional term}  hold and define $z'=z(\theta,\usub{1:G}^\prime)$ with $\usub{1:G}^\prime \sim \ov \pi(\cdot|\theta)$
and $z = z(\theta,\usub{1:k-1}^\prime, \usub{k}, \usub{k+1:G}^\prime) $ with $\usub{k} \sim p_{\Usub{k}} $ and independent of $\usub{1:G}^\prime$. Let $\rho = 1-1/G$. Then,
\begin{align*}
\begin{pmatrix}
z^\prime \\
z\end{pmatrix} & \sim  \N \begin{pmatrix} \begin{pmatrix} \frac12 \sigma^2 \\ -\frac12\sigma^2 (1-2\rho) \end{pmatrix} ;
 \sigma^2 \begin{pmatrix} 1 & \rho \\ \rho & 1
\end{pmatrix}
\end{pmatrix}.
\end{align*}
Hence,
${\Corr} (z,z')=\rho $.
\end{lemma}
\subparagraph{Pseudo-marginal based on $z$}\label{SS: PM based on z}
For the rest of the paper we will work with
 the MCMC scheme for $\theta$ and $z$ because the analysis of the original PM scheme based on blocking $\usub{1:G}$ is equivalent to that based $z$, but
 it is simpler to work with $(\theta, z) $. By Lemma~\ref{lemma: lemma on z's}, the target density for $(\theta,z)$ is $\ov \pi(\theta, z):=\exp(z)g_Z(z|\theta) \pi(\theta)$, with the proposal density
 for $z$ conditional on $z^\prime$ given by $\N \bigg  (z; - \frac{\sigma^2}{2G}+ \rho z'  , \sigma^2(1-\rho^2)\bigg )$.

Suppose that we are interested in estimating $\pi(\varphi)=\int\varphi(\t)\pi(\t)d\t$ for some scalar-valued function $\varphi(\t)$ of $\theta$.
Let $\{\theta^{[j]}, z^{[j]},j=1, \dots, M\}$ be the draws obtained from the PM sampler after it has converged,
and let the estimator of $\pi(\varphi)$ be  $\wh\pi(\varphi): =\frac{1}{M}\sum\varphi(\theta^{[j]})$.
We define the inefficiency of the estimator $\wh\pi(\varphi)$ relative to an estimator based
 on an i.i.d. sample from $\pi(\theta)$  as
 \begin{align} \label{eq: IF ratio}
\IF(\varphi, \sigma,\rho) : = \lim_{M \to \infty} M \Var_\text{PM}(\wh\pi(\varphi))/\Var_\pi(\varphi),
\end{align}
where $\Var_\text{PM}(\wh\pi(\varphi))$ is the variance of the estimator $\wh \pi(\varphi)$
 and $\Var_\pi(\varphi):=\E_\pi(\varphi(\t)^2)-[\E_\pi(\varphi(\theta))]^2$ so that $\Var_\pi(\varphi)/M$ is the variance  of the ideal estimator when $\theta^{[j]}\stackrel{iid}{\sim}\pi(\theta)$.
Lemma \ref{lem: IF} in Appendix~\ref{Proofs} shows that under our assumptions the inefficiency $\IF(\varphi, \sigma,\rho) $
is independent of $\varphi$ and is a function only of $\sigma$ and $\rho=1-1/G$. We write it as
$\IF(\sigma,\rho)$ and call it the inefficiency of the PM algorithm, and is a function of $\sigma$ for a given $\rho$.

Similarly to \cite{Pitt:2012}, we define the computing time of the sampler as
\begin{align} \label{eq: CT}
\CT(\sigma,\rho):=\frac{\IF(\sigma,\rho)}{\sigma^{1/\varpi}}.
\end{align}
This definition takes into account the total number of samples needed to obtain a given precision and the mixing rate of the PM chain.
It is justified in Appendix~\ref{app: CT derivation}.

To simplify the notation in this section we often do not show dependence on $\rho$ as it is assumed constant.
In Section~\ref{sec: large sample} we show that if we take $G = O(T^\half)$, then $\rho = 1 - O(T^{-\half})$ and  $N_k = O(T^{1/(4\varpi)})$ are optimal.
The next lemma shows the optimal $\sigma$ under our assumptions as well as the corresponding acceptance rates.
A similar result was previously obtained by \cite{deligiannidis:doucet:pitt:2015} for the correlated PM using MC, i.e., with $\varpi=1/2$.
\begin{lemma}[Optimally tuning BPM]\label{lem: sigma opt}
Suppose that Assumptions~\ref{ass: assumption on variances}-\ref{ass: perfect proposal},
 hold and
 $\rho = 1-1/G$ is fixed and close to 1. Then, the optimal $\sigma$ that minimizes $\CT(\s,\rho)$ is $\s_\opt\approx 2.16/\sqrt{1-\rho^2}$ if $\varpi=1/2$,
and $\sigma_\opt\approx 0.82/\sqrt{1-\rho^2}$ if $\varpi=3/2-\epsilon$ for any arbitrarily small $\eps>0$.
The unconditional acceptance rates (see \eqref{eq:accept uncontional} in the Appendix) under this optimal choice of the tuning parameters are approximately 0.28 (MC) and 0.68, (RQMC) respectively.
\end{lemma}
Let $M$ be the length of the generated Markov chain.
The average number of times that a block $u_{(k)}$ is updated is $M/G$.
In general, $G$ should be selected such that $M/G$ is not too small
so that the space of $z$ is adequately explored.
In the examples in this paper, if not otherwise stated, we set $G=100$,
as we found that the efficiency is relatively insensitive to larger values of $G$.
Lemma \ref{lem: sigma opt} states that if the likelihood is estimated by MC, i.e. $\varpi=1/2$,
then the  optimal variance of the log-likelihood estimator based on each group is
$\sigma_\opt^2/G\approx 2.16^2/(1+\rho)$, which is approximately $ 2.34$ given that $G \approx 100 $
 is large.
For RQMC, the optimal variance of the log-likelihood estimator based on each group is
$\sigma_\opt^2/G\approx 0.82^2/(1+\rho)\approx 0.34$, given that $G $ is large.
Hence, for each group $k$, we propose tuning the number of samples $N_{(k)}=N_{(k)}(\theta)$ such that
$\Var(z_{(k)}|\t,N_{(k)})$ is approximately $2.34$ if the likelihood is estimated by MC
or  $0.34$
if it is estimated by  RQMC.
In many cases, it is more convenient to tune $N_{(k)}=N_{(k)}(\bar\t)$
at some central value $\bar\t$ and then fix $N_{(k)}$ across all MCMC iterations.

\section{Applications}\label{sec:applications}
This section illustrates the methodology with three applications. Appendix~\ref{S: further applications} gives two further applications to Approximate Bayesian Computation (ABC) and to non-Gaussian state space models.
\subsection{Panel-data example\label{sub:PanelData}}

A clinical trial is conducted to test the effectiveness of beta-carotene in preventing
non-melanoma skin cancer \citep{Greenberg:1989}.
Patients were randomly assigned to a control or treatment group
and biopsied once a year to ascertain the number of new skin cancers since the last examination.
The response $y_{ij}$ is a count of the number of new skin cancers in year $j$ for  patient $i$.
Covariates include age, skin (1 if skin has burns and 0 otherwise), gender, exposure (a count of the number of previous
skin cancers), year of follow-up and treatment (1 if
the patient is in the treatment group and 0 otherwise).
There are $T = 1683$ patients with complete covariate information.
We follow \cite{Donohue:2011} and consider the mixed Poisson model with a random intercept
\begin{align*}
p(y_{ij}|\beta,\alpha_i)&=\text{Poisson}(\exp(\eta_{ij})),\quad
\eta_{ij}=\beta_0+\beta_1\text{Age}_{i}+\beta_2\text{Skin}_{i}+\beta_3\text{Gender}_{i}+\b_4\text{Exposure}_{ij}+\alpha_{i},
\end{align*}
where $\alpha_i\sim \N(0,\varrho^2)$, $i = 1,....,T=1683$, $j=1,...,n_i=5$.
The likelihood is
\beqn
L(\t)=\prod_{i=1}^T L_i(\t),\;\; L_i(\t):=p(y_i|\t)=\int\left(\prod_{j=1}^{n_i}p(y_{ij}|\beta,\alpha_i)\right)p(\alpha_i|\varrho^2)d\alpha_i
\eeqn
with $\t=(\beta,\varrho^2)$ the vector of the unknown parameters of the model.

We ran both the optimal independent PM and the optimal block PM for 50,000 iterations, with the first 10,000 discarded as burn-in. We do not compare BPM
to CPM here since it is not clear how to evolve the random numbers when the $N_i$ vary over the iterations. For a fixed sample size for each $i$, we will show in Section~\ref{SS: MC vs RQMC} and, in particular, in Table~\ref{tab:quasi_pseudo}, that block PM performs
 better than correlated PM.
In all our examples the likelihood is estimated using MC using pseudo random numbers if not otherwise stated.
For simplicity, each likelihood $L_i(\t)$ is estimated by importance sampling based on $N_i$ i.i.d. samples from the natural importance sampler $p(\a_i|\varrho^2)$.
For the independent PM, for each $\theta$, the number of samples $N_i=N_i(\theta)$ is tuned so as the variance of the log-likelihood estimator $\Var(\log\wh L_i(\theta))$ is not bigger than $1/T$ (to target the optimal variance of $1$ for the log-likelihood).
This is done as follows.
We start from some small $N_i$ and increase $N_i$ if this variance is bigger than $1/T$.
We note that an explicit expression is available for an estimate of the variance $\Var(\log\wh L_i(\theta))$.
The CPU time spent on tuning $N_i$ is taken into account in the comparison.
In the block PM, we divide the data into $G=99$ groups,
so that each group has 17 panels,
and the variance of the  log-likelihood estimator in each group is tuned to not be bigger than the optimal value of $2.34$; see Lemma~\ref{lem: sigma opt}
and the discussion following it.

As performance measures, we report the acceptance rate,
the integrated autocorrelation time (IACT),
the CPU times, and the time normalised variance (TNV).
For a univariate parameter $\t$, the IACT is estimated by
\beqn
\text{IACT}=1+2\sum_{t=1}^{1000}\widehat\rho_t,
\eeqn
where $\widehat\rho_t$ are the sample autocorrelations.
For a multivariate parameter, we report the average of the estimated IACT's.
The time normalised variance is the product of the IACT and the CPU time
The TNV is proportional to the computing time defined in \eqref{eq: CT}
if the CPU time to generate $N$ samples is proportional to $N$.

Table~\ref{tab:skin cancer} summarises the acceptance rates,
the IACT ratio, the CPU ratio, and the TNV ratio,
using the block PM as the baseline.
The table shows that the block PM outperforms the independent PM.
In particular, the block PM is around 25 times more efficient than the independent PM in terms of the time normalised variance.


\begin{table}[h]
\centering
\vskip2mm
{\small
\begin{tabular}{ccccc}
\hline 
Methods	&Acceptance rate 	&IACT ratio	&CPU ratio	&TVN ratio\\
\hline
IPM&0.222	&1.080		& 23.095		& 24.938\\
BPM&0.243	&1		& 1		& 1 \\
\hline
\end{tabular}
}
\caption{Panel-data example: Comparison the block PM and independent PM  using the block PM as the baseline.}\label{tab:skin cancer}
\end{table}

\subsubsection{Optimally choosing a static number of samples $N$}\label{SS: choosing a static N optimally}
In other applications it may be more costly to select the optimal numbers of samples,  $N_i$,  to estimate $L_i(\theta)$ for any $\theta$ in each PM iteration. We will now investigate the performance of a more easily implemented and less costly static strategy where the $N_i$ are fixed across $\theta$
and are tuned at a central $\bar\theta$ obtained by a short pilot run. We would like to verify that Lemma \ref{lem: sigma opt} still provides a sensible strategy for selecting such a static number of samples.
Because there are 17 panels in each group, given a target group-variance $\sigma^2_G$,
$N_i=N_i(\bar\theta)$ is selected such that
$\Var(\log \wh L_i(\bar\theta))\approx \s^2_G/17$, so that $\Var(\zsub{k})\approx \sigma^2_G$.

Figure \ref{fig:optimal N} shows the average $\bar N=\sum N_i/T$,
IACT and computing time CT = $\bar N\times$ IACT for various group-variance $\sigma^2_G$, when the likelihood is estimated using MC.
The computing time CT is minimised at $\sigma^2_G\approx 2.3$,
which requires 40 samples on average to estimate each $L_i(\theta)$.
In this example we found that CT does not change much when $\sigma^2_G$ lies between 2 and 2.4.
The computing time increases slowly when we choose the $N_i$ such that $\s^2_G$
decreases from its optimal value, but increases dramatically when $\s^2_G$ increases from its optimal value.
To be on the safe side, we therefore advocate a conservative choice of $N_i$ in practice.

\begin{figure*}[h]
\centerline{\includegraphics[width=1\textwidth,height=.4\textwidth]{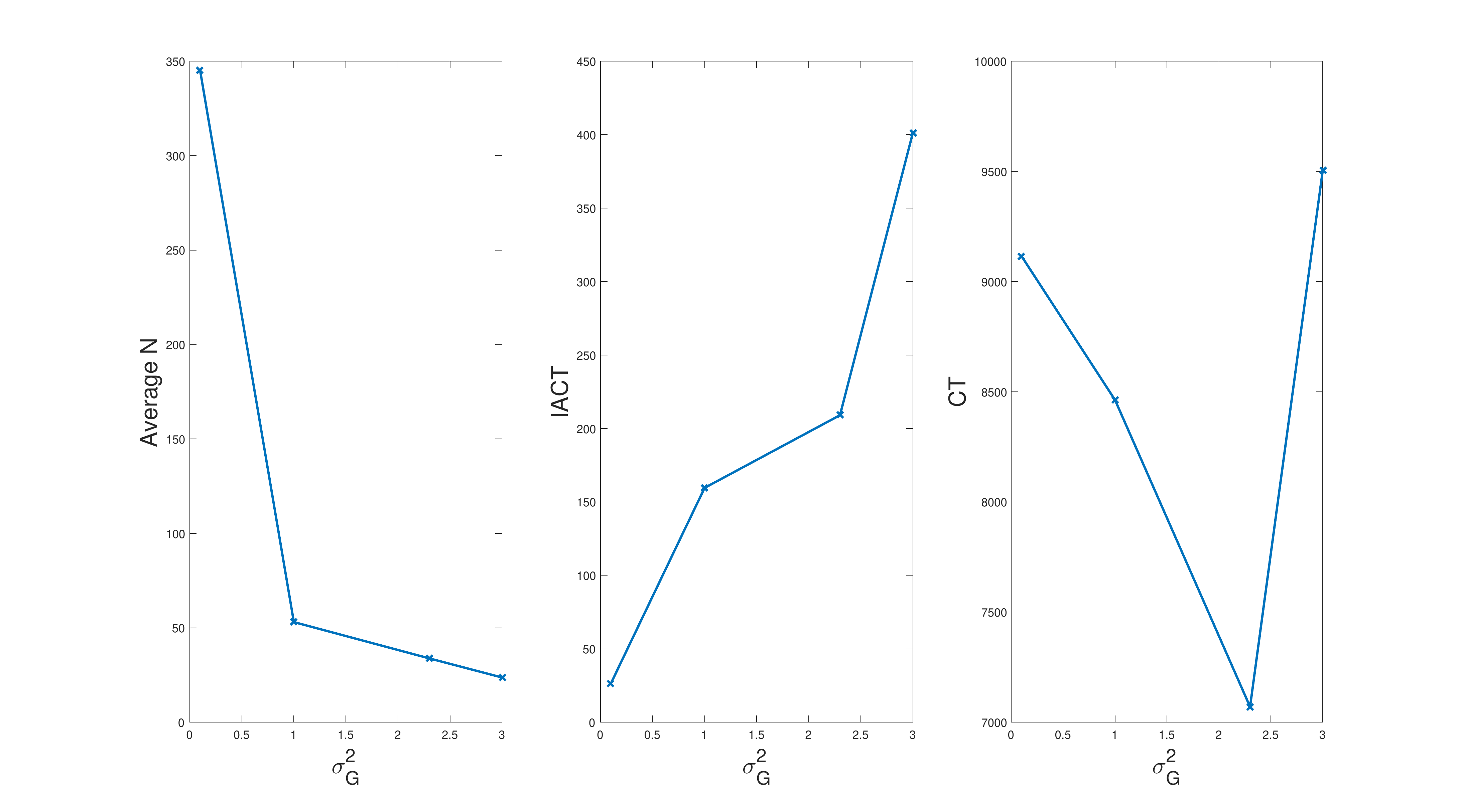}}
\caption{\label{fig:optimal N}
Panel-data example: Average $N_i$ ($\bar N$), IACT and CT = $\bar N\times$ IACT for various target $\sigma_G^2$, using MC.
}
\end{figure*}

We now report results using RQMC to estimate the likelihood,
using the scrambled net algorithm of \cite{Matousek:1998}.
We note that if $L_i(\theta)$ is estimated using RQMC, then the generated scrambled quasi random numbers are dependent
although the estimate $\wh L_i (\theta)$ is still unbiased.
Thus, unlike MC, it is difficult to obtain a closed form expression for an unbiased estimator of the variance of $\wh L_i (\theta)$.
We therefore use replication to estimate each $\Var(\log \wh L_i(\bar\t))$.
Figure \ref{fig:optimal N quasi} shows that CT is minimised at $\s^2_G\approx 0.3$,
which agrees with the theory in Lemma~\ref{lem: sigma opt}.
CT increases slowly when $\s^2_G$ is smaller 0.3, but it increases quickly when $\s^2_G$ is higher than this value.

\begin{figure*}[h]
\centerline{\includegraphics[width=1\textwidth,height=.4\textwidth]{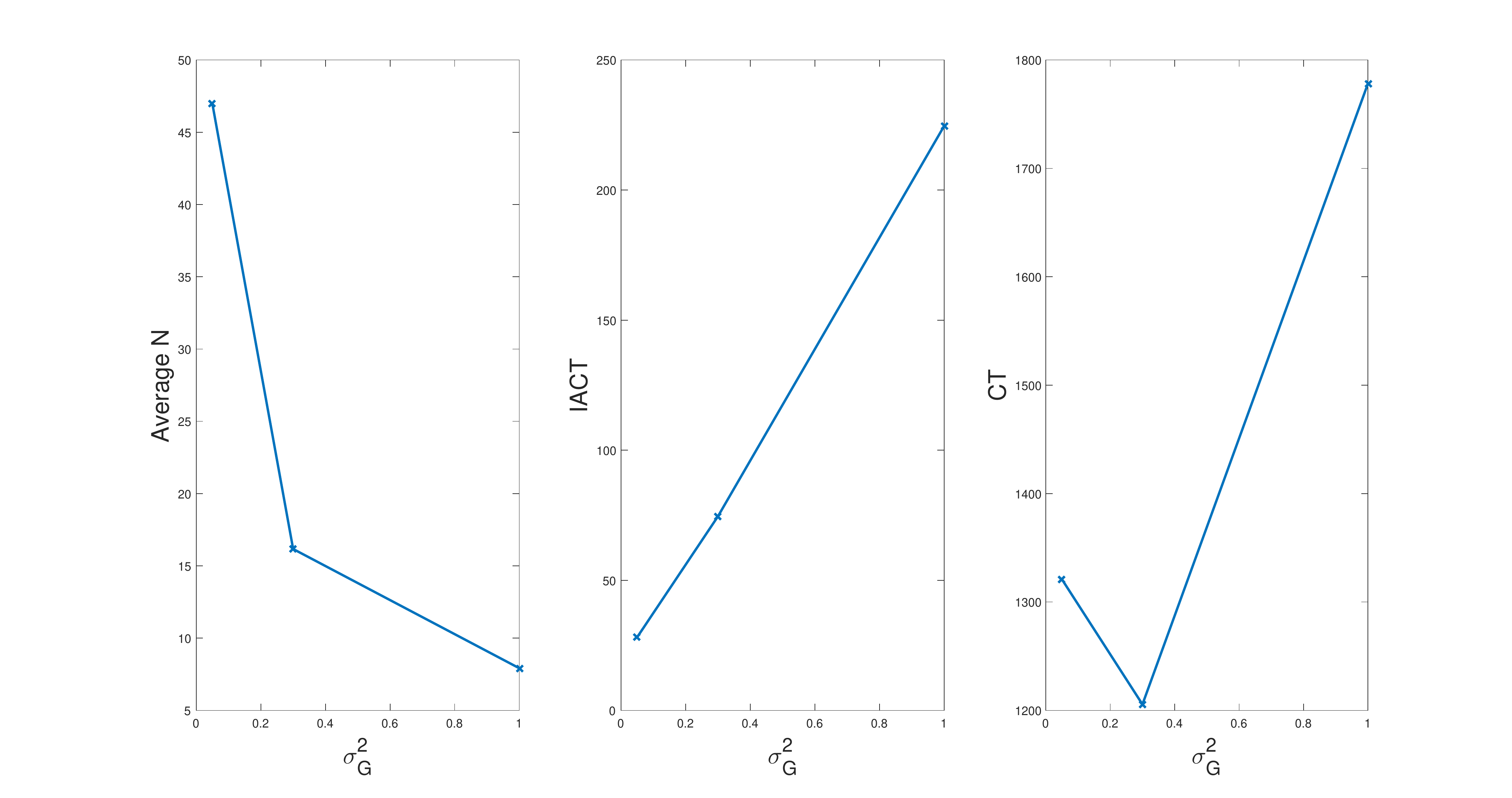}}
\caption{\label{fig:optimal N quasi}
Panel-data example: Average $N_i$ ($\bar N$), IACT and CT = $\bar N\times$ IACT for various $\s_G^2$, using RQMC.
}
\end{figure*}

\subsubsection{MC vs RQMC}\label{SS: MC vs RQMC}
We now compare the performance of the various schemes, and for simplicity use the same number of samples $N$ in all methods. We consider four schemes: independent PM using MC (IPM-MC), correlated PM using MC (CPM-MC), block PM using MC (BPM-MC), and block PM using RQMC (BPM-RQMC). We set $N_i=50$ across all $\theta$ and $i$, and verified that this was enough for both CPM and BPM chains to converge. The IPM-MC chain is unlikely to converge in this setting, as it requires a much larger $N$.

Table~\ref{tab:quasi_pseudo} summarises the performance measures using the BPM-RQMC as the baseline.
The two block PM frameworks outperform both IPM-MC and CPM-MC. BPM-MC is somewhat faster than BPM-RQMC, but BPM-RQMC is much more efficient and has three times lower TNV compared to BPM-MC.

\begin{table}[h]
\centering
\vskip2mm
{\small
\begin{tabular}{ccccc}
\hline 
Methods						&Acceptance rate 	&IACT ratio	&CPU ratio	&TNV ratio\\
\hline
IPM-MC				&0.002			&12.005		&1.124    		& 13.493\\
CPM-MC				&0.081			&5.273		&1.133    		& 5.974\\
BPM-MC	      &0.179 			&4.121 		&0.742		     &3.057 \\
BPM-RQMC		  &0.225				&1				&1				& 1 \\
\hline
\end{tabular}
}
\caption{Panel-data example: comparison of independent PM using MC, block PM using MC, and block PM using RQMC. The BPM-RQMC is used as the baseline. }\label{tab:quasi_pseudo}
\end{table}


\subsection{Data subsampling example}\label{SS: data subsampling}
\cite{quiroz:villani:kohn::2016} propose a data subsampling approach to Bayesian inference to speed up MCMC when the likelihood  can be computed.
The subsampling approach expresses the log-likelihood as a sum of terms and estimates it unbiasedly by summing a sample of the terms using
control variates and simple random sampling. The unbiased log-likelihood estimator is converted to a slightly biased likelihood estimator in \cite{quiroz:villani:kohn::2016} such that the PM targets a slightly perturbed target posterior. See also \cite{Quiroz:2016} for an alternative unbiased estimator. \cite{quiroz:villani:kohn::2016}
use both the correlated PM and the block PM
to carry out the estimation. For the block PM, $N_{(k)} = N_{k} =  T/G$.

We illustrate the subsampling approach of \cite{quiroz:villani:kohn::2016} and compare the block PM to the correlated PM using the following two $\mathrm{AR}(1)$ models with Student-t iid errors $\epsilon_{t}\sim t(\nu)$ with known degrees of freedom $\nu$. These examples are also used in their paper. The two models are
$
{\rm M}_1:\enspace y_{t}  = \beta_{0}+\beta_{1}y_{t-1}+\epsilon_{t},q$ with $ \theta=(\beta_{0}=0.3,\beta_{1}=0.6) $ and
${\rm M}_2:\enspace y_{t}  = \mu+\rho(y_{t-1}-\mu)+\epsilon_{t}$, with $  \theta=(\mu=0.3,\rho=0.99)$,
for $t=1,...,T$, where $p(\epsilon_{t})\propto(1+\epsilon_{t}^{2}/\nu)^{-(\nu+1)/2}$
with $\nu=5$. Our aim is not to compare the two models, but to investigate the behaviour of BPM and CPM when data are generated from respective model.
We use the same priors as in \cite{quiroz:villani:kohn::2016}:
$  M_1: p(\beta_{0},\beta_{1})=\mathcal{U}(-5,5)\cdot\mathcal{U}(0,1) \quad \text{and} \quad   M_2: p(\mu,\rho)=\mathcal{U}(-5,5)\cdot\mathcal{U}(0,1)$,
where $\mathcal{U}(a,b)$ means a uniform density on the interval $(a,b)$.

Define $\ell_t(\theta):=\log p(y_t | y_{t-1}, \theta)$ and rewrite the log-likelihood $\ell(\theta)$ as
\begin{align*}
\ell(\theta) & = q(\theta) + d(\theta), \quad q(\theta)=\sum_{t=1}^T q_t(\theta), \quad d(\theta)=\sum_{t=1}^T  d_t(\theta), \text{ with } d_t(\theta)= \ell_t(\theta)-q_t(\theta),
\end{align*}
where $q_t(\theta)\approx \ell_t(\theta)$ is a control variate. We take $q_{t}(\theta)$ as a second order Taylor series approximation of $l_{t}(\theta)$ evaluated at the nearest centroid from a clustering in data space. This reduces the complexity of computing $q(\theta)$ from $O(T)$  to $O(C)$, where $C$ is the number of centroids.
See \cite{quiroz:villani:kohn::2016} for the details.
An unbiased estimate of $\ell(\theta)$ based on a simple random sample with replacement is
\begin{align}
\widehat{\ell}(\theta) & = \wh d(\t)+ q(\theta),\label{eq:loglikelihoodEstsubsampling}
\end{align}
where
\[\wh d(\t)=\frac{T}{N}\sum_{i=1}^{N}d_{u_{i}}(\theta), \quad \text{with }\;  u_i \in \{1, \dots, T\},\;\;\P(u_i = t) = \frac{1}{T},\quad t=1,\dots, T.\]
Here $N$ is the subsample size and $u=(u_1,...,u_N)$ represents a vector of observation indices.
Write $\wh{d}(\t)$ as a sum of $G$ blocks
\[
\wh{d}=\wh{d}^{(1)}+\cdots+\wh{d}^{(G)},\quad\text{with}\quad \wh{d}^{(k)}=\frac{T}{N}\sum_{i\in\mathcal{I}_{k}}d_{u_{i}},
\]
where $\mathcal{I}_{k}$ with $|\mathcal{I}_{k}|=N_{(k)}$ contains
the indices of the auxiliary variables corresponding to the $k$th block.
We assume that the $N_{(k)}$ are the same for all $k$ and $N=G\times N_{(k)}$.
Let $\s^2(\t)=\Var(\wh l(\t))=(T/N)\sum_{t=1}^T(d_i(\t)-\bar d(\t))^2$ with $\bar d(\t)=\sum d_i(\t)/T$.
Notice that $\E[\wh{d}^{(k)}]=d(\theta)/{G}$ and $\Var[\wh{d}^{(k)}]={\sigma^{2}}/{G}$.
Using the result that if $\wh d\sim\N(d,\sigma^2/2)$, we have that
 $\E[\exp(q(\theta)+\wh d(\theta)-\sigma^2(\theta)/2)]=\exp(l(\theta))$,
\cite{quiroz:villani:kohn::2016} work with the likelihood estimate
\begin{align}\label{eq: approx estimate of likelihood}
\wh L(\t,u)&=\exp(q(\theta))\exp\left(\wh{d(\theta)}-\frac{\wh\sigma^{2}(\theta)}{2}\right) =
  \exp(q(\theta))\prod_{k=1}^{G}\exp\left(\wh{d(\theta) }^{(k)}-\frac{\wh\sigma^{2}(\theta)}{2G}\right),
\end{align}
where $\widehat{\sigma}^2(\theta)$ is an unbiased estimate of $\s^2(\t)$, because computing $\sigma^2(\theta)$, to obtain an unbiased estimator of the likelihood,
is expensive and defeats the purpose of subsampling.
\cite{quiroz:villani:kohn::2016} show that carrying out the PM with this slightly biased
likelihood estimator samples from a perturbed posterior that is very close to the full-data posterior under quite general conditions.

We generated $T= 100, 000$ observations from the models in ${\rm M}_1 $ and ${\rm M}_2 $ and ran both the correlated PM and the block PM for $55,000$ iterations from which we discarded the first 5,000 draws as burn-in. Using the same target for $\sigma^2(\theta)$
as in \cite{quiroz:villani:kohn::2016} results in sample sizes $N\approx 1300$ for model $M_1$ and $N\approx 2600$ for model $M_2$.
For the block PM we use $G=100$. Also, following  \cite{quiroz:villani:kohn::2016}, the correlation parameter in the
correlated PM is set to $\varrho=0.9999$, and we use a random walk proposal which is adapted during the burn-in phase to target an
 acceptance rate of approximately $0.15$ \citep{Sherlock:2015}.

Table~\ref{tab:sub sampling} summarises the performance measures introduced in Section~\ref{sub:PanelData}. It is evident that the block
PM significantly outperforms the correlated PM in terms of CPU time and TNV. This is because, as discussed above, the correlated
PM requires $N$ operations for generating the vector $\bm u$. The block PM moves only one block at a time, so that the update of the vector $u$ requires $N/G$ operations.

\begin{table}[h]
\centering
\vskip2mm
{\small
\begin{tabular}{ccccccccc}
  \hline
   Methods & \multicolumn{2}{r}{Acceptance rate} & \multicolumn{2}{c}{IACT ratio} & \multicolumn{2}{c}{CPU ratio} & \multicolumn{2}{c}{TNV ratio} \\ \hline
  ~   & $\mathrm{M}_1$ & $\mathrm{M}_2$ & $\mathrm{M}_1$ & $\mathrm{M}_2$ & $\mathrm{M}_1$ & $\mathrm{M}_2$ & $\mathrm{M}_1$ & $\mathrm{M}_2$ \\ 
    CPM    & 0.149 & 0.140 & 1.110 & 1.124 & 62.893 & 38.610 & 69.444 & 43.478 \\ 
    BPM    & 0.160 & 0.151 & 1 & 1 & 1 & 1 & 1 & 1 \\ 

 \hline
\end{tabular}
}
\caption{Data subsampling example using block PM as a baseline.}\label{tab:sub sampling}
\end{table}

\subsection{Diffusion process example}\label{sub:Diffusion process example}
This section applies the PM approaches to estimate the parameters of the  diffusion process
$\bm X=\{X_t,t\geq 0\}$ governed by the stochastic differential equation (SDE)
\beq\label{eq:SDE}
dX_t=\mu(X_t,\t)dt+\s(X_t,\t)dW_t,
\eeq
with $W_t$ a Wiener process. We assume that the regularity conditions on $\mu(\cdot,\cdot)$ and $\sigma(\cdot,\cdot)$ are met
so that the solution to the SDE in \eqref{eq:SDE} exists and is unique. We are interested in estimating the vector of parameters $\theta$
based on discrete-time observations $x=\{x_0,x_1,...,x_n\}$, where $x_i$ is the observation of $X_{i\Delta}$ with $\Delta$ some time-interval. The likelihood is
\[L(\t):=p(x|\t)=\prod_{i=0}^{n-1}p_\Delta(x_{i+1}|x_i,\t),
\]
 where the $\Delta$-interval Markov transition density $p_\Delta(x_{i+1}|x_i,\t)$ is typically intractable.
 In order to make the discrete approximation of the continuous-time process $X$ sufficiently accurate,
we follow  \cite{Stramer:2011} and write $p_\Delta(x_{i+1}|x_i,\theta)$ as
\beq\label{eq:p_delta}
p_\Delta(x_{i+1}|x_i,\t)=\int p_h(x_{i+1}|z_{i,M-1},\t)p_h(z_{i,M-1}|z_{i,M-2},\t)\cdots p_h(z_{i,1}|x_{i},\t)\d z_{i,1}\cdots \d z_{i,M-1},
\eeq
where  $p_h(\cdot|\cdot,\theta)$ is the Markov transition density of $X$ after time-step $h=\Delta/M$.
The Euler approximation
\[p_h^\text{euler}(u|v,\t)=\N\Big(u;v+h\mu(v,\t),h\Sigma(v,\t)\Big),\]
with $\Sigma(v,\t)=\s(v,\t)\s'(v,\t)$, is a very accurate approximation to
$p_h(u|v,\theta)$  if $h$  is sufficiently small.
We approximate the transition density in \eqref{eq:p_delta} by
 \[p_\Delta^\text{euler}(x_{i+1}|x_i,\t)=\int p_h^\text{euler}(x_{i+1}|z_{i,M-1},\t)p_h^\text{euler}(z_{i,M-1}|z_{i,M-2},\t)\cdots p_h^\text{euler}(z_{i,1}|x_{i},\t)\d z_{i,1}\cdots \d z_{i,M-1},\]
and follow \cite{Stramer:2011} and define the working likelihood as
\[L^\text{euler}(\t)=\prod_{i=0}^{n-1}p_\Delta^\text{euler}(x_{i+1}|x_i,\t).\]
The posterior density of $\theta$ is then
$p^\text{euler}(\theta|x)\propto p_\Theta(\t)L^\text{euler}(\t).$
The likelihood $L^\text{euler}(\t)$ is intractable, but can be estimated unbiasedly.
As in \cite{Stramer:2011}, we estimate $p_\Delta^\text{euler}(x_{i+1}|x_i,\t)$  using the importance sampler of \cite{Durham:2002},
\[z_{i,m+1}\sim \N\left(z_{i,m}+\frac{x_{i+1}-z_{i,m}}{M-m},h\frac{M-m-1}{M-m}\Sigma(z_{i,m},\t)\right),\;\;m=0,...,M-2,\]
where $z_{i,0}=x_i$.
The density of this importance distribution is
\[g(z_i)=g(z_{i,1},...,z_{i,M-1})=\prod_{m=0}^{M-2}\N\left(z_{i,m+1};z_{i,m}+\frac{x_{i+1}-z_{i,m}}{M-m},h\frac{M-m-1}{M-m}\Sigma(z_{i,m},\t)\right).\]
We sample $N$ such trajectories $z_{i}^{(j)}=(z_{i,1}^{(j)},...,z_{i,M-1}^{(j)})$, $j=1,...,N$
and denote by $\bm u_i$ the set of all required MC random numbers, $i=0,...,n-1$.
Then, the unbiased estimator of $p_\Delta^\text{euler}(x_{i+1}|x_i,\t)$ is
 \[\wh p_\Delta^\text{euler}(x_{i+1}|u_i,x_i,\t)=\frac1N\sum_{j=1}^N\frac{p_h^\text{euler}(x_{i+1}|z_{i,M-1}^{(j)},\t)p_h^\text{euler}(z_{i,M-1}^{(j)}|z_{i,M-2}^{(j)},\t)\cdots . p_h^\text{euler}(z_{i,1}^{(j)}|x_{i},\t)}{g(z_{i}^{(j)})}
 \]
The working likelihood $L^\text{euler}(\t)$ factorises as in \eqref{eq:panel data llh} and is estimated unbiasedly,
so all the theory developed in Sections~\ref{SS: analysis of block} and \ref{sec: large sample} applies here as well.

We apply the proposed method to fit the FedFunds dataset to the Cox-Ingersoll-Ross (CIR) model
\[dX_t=\beta(\a-X_t)dt+\s\sqrt{X_t}dW_t, \]
using MC pseudo random numbers.
The FedFunds dataset we use consists of 745 monthly federal funds rates in the US from July 1954 to August 2016,
downloaded~from~Yahoo~Finance  (https://au.finance.yahoo.com/).

We follow \cite{Stramer:2011} and set $\Delta=1/12$ and also use the prior
\begin{align*}
p_\T(\t)=I_{(0,1)}(\a)I_{(0,\infty)}(\b)\s^{-1}I_{(0,\infty)}(\s)
\end{align*}
where $I_{(a,b)}(x) = 1$ if $x \in (a,b)$ and 0 otherwise.

We take $M=300$ to make the Euler approximation highly accurate, \cite{Stramer:2011} use $M=20$.
We use $N=1$ samples and $G=186$ groups so that $u_{(k)}=\{u_{3(k-1)},u_{3(k-1)+1},u_{3(k-1)+2}\}$, $k=1,...,G$.
Table \ref{tab:fedfunds} summarises the results,
which show that the block PM performs better than the independent PM.
\cite{Stramer:2011} report that their blocking strategy does not work better than the independent PM.
There are three reasons for this different conclusion. First, \citeauthor{Stramer:2011}'s dataset consists of 432 monthly rates from January 1963 to December 1998,
which is a little over half of our dataset. Second, they set $M=20$ and $N=5$ while we set $M=300$ and $N=1$.
For both these reasons, the estimate of the log likelihood in their problem has a variance that is small and less than 1 and hence our theory
predicts that the independent PM will be as good as the block PM, and shows the value of our theoretical guidelines. The variance of the log of the likelihood
estimate in our setting is much greater than 1 so that our setting is much more challenging for the independent PM
because the estimates of the likelihood are highly variable.
Third, \cite{Stramer:2011} use a MCMC scheme that treats $\theta$ and the $G$ blocks
$\usub{1}$,...,$\usub{G}$ as $G+1$ blocks that are generated one at a time conditional on all the other blocks.
Our MCMC scheme updates $\theta$ and one of the $\usub{i}$ jointly in each iteration.
\begin{table}[h]
\centering
\vskip2mm
{\small
\begin{tabular}{ccccc}
\hline 
Methods						&Acceptance rate	&IACT ratio	&CPU ratio	&TNV ratio\\
\hline
IPM				&0.049			&9.059		&1.154    		& 10.45\\
BPM				&0.258 			&1				&1				& 1 \\
\hline
\end{tabular}
}
\caption{Diffusion process example: Comparing the independent PM (IPM) and the block PM (BPM)  using the block PM as the baseline.}\label{tab:fedfunds}
\end{table}

\section{Conclusion}\label{Sec:Conclusion}

\cite{deligiannidis:doucet:pitt:2015} show how the PM approach can be made much more
efficient by correlating the underlying Monte Carlo (MC) random numbers used to form the estimate of the likelihood at the current and
proposed values of the unknown parameters. Their approach greatly speeds up the standard PM algorithm,
as it requires a much smaller number of samples or particles to form the optimal likelihood estimate.
Our paper presents an alternative implementation of the correlated PM approach, called the
block PM, which divides the underlying random numbers into blocks so that the likelihood estimates for the
proposed and current values of the parameters only differ by the random numbers in one block.
We show that this implementation of the correlated PM can be much more efficient for some specific problems
than the implementation in \cite{deligiannidis:doucet:pitt:2015}; for example
when the likelihood is estimated by subsampling or the likelihood is a product of terms each of which is given by an integral which
can be estimated unbiasedly by randomised quasi-Monte Carlo. Using stylized but realistic assumptions
the article also provides methods and guidelines for implementing the block PM efficiently.
As already discussed, we have successfully implemented the block PM in several  applications and shown that it results
in greatly improved performance of the PM sampler.
A second advantage of the the block PM is that it provides a direct way to control the correlation between the logarithms of the estimates of the
likelihood at the current and proposed values of the parameters than the implementation in \cite{deligiannidis:doucet:pitt:2015}.
We obtain methods and guidelines for selecting the optimal number of samples based on idealized but realistic assumptions.
Finally, we believe that in future applications CPM can be combined with BPM to produce efficient PM algorithms.

\section*{Acknowledgement} We would like to thank Mike Pitt for useful discussions and in particular a  version of
Lemma~\ref{lemma:uncorrelated in n}.
Robert Kohn and Matias Quiroz were partially supported
by an Australian Research Council Center of Excellence Grant CE140100049.
Villani was partially supported by Swedish Foundation for Strategic Research (Smart Systems: RIT 15-0097)

\bibliographystyle{apalike}
\bibliography{references}

\begin{thebibliography}{}

\bibitem[Andrieu et~al., 2010]{Andrieu:2010}
Andrieu, C., Doucet, A., and Holenstein, R. (2010).
\newblock Particle {Markov chain Monte Carlo} methods.
\newblock {\em Journal of the Royal Statistical Society, Series B}, 72:1--33.

\bibitem[Andrieu and Roberts, 2009]{Andrieu:2009}
Andrieu, C. and Roberts, G. (2009).
\newblock The pseudo-marginal approach for efficient {Monte Carlo}
  computations.
\newblock {\em The Annals of Statistics}, 37:697--725.

\bibitem[Bartolucci et~al., 2012]{Bartolucci:2012}
Bartolucci, F., Farcomeni, A., and Pennoni, F. (2012).
\newblock {\em Latent Markov Models for Longitudinal Data}.
\newblock Chapman and Hall/CRC press.

\bibitem[Beaumont, 2003]{Beaumont:2003}
Beaumont, M.~A. (2003).
\newblock Estimation of population growth or decline in genetically monitored
  populations.
\newblock {\em Genetics}, 164:1139--1160.

\bibitem[Bornn et~al., 2016]{Bornn2016}
Bornn, L., Pillai, N.~S., Smith, A., and Woodard, D. (2016).
\newblock {The use of a single pseudo-sample in approximate Bayesian
  computation}.
\newblock {\em Statistics and Computing}, pages 1--8.

\bibitem[Dahlin et~al., 2015]{Dahlin:2015}
Dahlin, J., Lindsten, F., Kronander, J., and Sch\"{o}n, T.~B. (2015).
\newblock Accelerating pseudo-marginal {Metropolis-Hastings} by correlating
  auxiliary variables.
\newblock Technical report.
\newblock https://arxiv.org/abs/1511.05483.

\bibitem[Del~Moral, 2004]{DelMoral:2004}
Del~Moral, P. (2004).
\newblock {\em Feynman-Kac Formulae: Genealogical and Interacting Particle
  Systems with Applications}.
\newblock Springer, New York.

\bibitem[Deligiannidis et~al., 2016]{deligiannidis:doucet:pitt:2015}
Deligiannidis, G., Doucet, A., and Pitt, M. (2016).
\newblock The correlated pseudo-marginal method.
\newblock Technical report.
\newblock http://arxiv.org/abs/1511.04992v3.

\bibitem[Dick and Pillichshammer, 2010]{Dick2010}
Dick, J. and Pillichshammer, F. (2010).
\newblock {\em Digital nets and sequence. Discrepancy theory and quasi-{Monte
  Carlo} integration}.
\newblock Cambridge University Press, Cambridge.

\bibitem[Donohue et~al., 2011]{Donohue:2011}
Donohue, M.~C., Overholser, R., Xu, R., and Vaida, F. (2011).
\newblock Conditional {A}kaike information under generalized linear and
  proportional hazards mixed models.
\newblock {\em Biometrika}, 98:685--700.

\bibitem[Doucet et~al., 2015]{Doucet:Pitt:Deligiannidis:Kohn}
Doucet, A., Pitt, M., Deligiannidis, G., and Kohn, R. (2015).
\newblock Efficient implementation of {Markov chain Monte Carlo} when using an
  unbiased likelihood estimator.
\newblock {\em Biometrika}, 102(2):295--313.

\bibitem[Durbin and Koopman, 1997]{Durbin:1997}
Durbin, J. and Koopman, S.~J. (1997).
\newblock {Monte Carlo} maximum likelihood estimation for {non-Gaussian} state
  space models.
\newblock {\em Biometrika}, 84:669--684.

\bibitem[Durham and Gallant, 2002]{Durham:2002}
Durham, G.~B. and Gallant, A.~R. (2002).
\newblock Numerical techniques for maximum likelihood estimation of
  continuous-time diffusion processes.
\newblock {\em Journal of Business \& Economic Statistics}, 20(3):335--338.

\bibitem[Fitzmaurice et~al., 2011]{Fitzmaurice:2011}
Fitzmaurice, G.~M., Laird, N.~M., and Ware, J.~H. (2011).
\newblock {\em Applied Longitudinal Analysis}.
\newblock John Wiley \& Sons, Ltd, New Jersey, 2nd edition.

\bibitem[Flury and Shephard, 2011]{Flury:2011}
Flury, T. and Shephard, N. (2011).
\newblock Bayesian inference based only on simulated likelihood: {P}article
  filter analysis of dynamic economic models.
\newblock {\em Econometric Theory}, 1:1--24.

\bibitem[Gerber and Chopin, 2015]{Gerber:2015}
Gerber, M. and Chopin, N. (2015).
\newblock Sequential quasi {M}onte {C}arlo.
\newblock {\em Journal of the Royal Statistical Society, Series B},
  77(3):509--579.

\bibitem[Gourieroux and Monfort, 1995]{Gourieroux:1995}
Gourieroux, C. and Monfort, A. (1995).
\newblock {\em Statistics and Econometric Models}, volume~2.
\newblock Cambridge University Press, Melbourne.

\bibitem[Greenberg et~al., 1989]{Greenberg:1989}
Greenberg, E.~R., Baron, J.~A., Stevens, M.~M., Stukel, T.~A., Mandel, J.~S.,
  Spencer, S.~K., Elias, P.~M., Lowe, N., Nierenberg, D.~N., G., B., and Vance,
  J.~C. (1989).
\newblock The skin cancer prevention study: design of a clinical trial of
  beta-carotene among persons at high risk for nonmelanoma skin cancer.
\newblock {\em Controlled Clinical Trials}, 10:153--166.

\bibitem[Gunawan et~al., 2016]{Gunawan:2016}
Gunawan, D., Tran, M.-N., Suzuki, K., Dick, J., and Kohn, R. (2016).
\newblock Computationally efficient {Bayesian} estimation of high dimensional
  copulas with discrete and mixed margins.
\newblock Technical report.
\newblock http://arxiv.org/abs/1608.06174.

\bibitem[Johnson et~al., 2013]{Johnson:2013}
Johnson, A.~A., Jones, G.~L., and Neath, R.~C. (2013).
\newblock Component-wise {Markov chain Monte Carlo}: Uniform and geometric
  ergodicity under mixing and composition.
\newblock {\em Statistical Science}, 28(3):360--375.

\bibitem[Lee and Holmes, 2010]{Lee:2010}
Lee, A. and Holmes, C. (2010).
\newblock Discussion on particle {Markov chain Monte Carlo} methods.
\newblock {\em Journal of the Royal Statistical Society, Series B}, 72:1--33.

\bibitem[Lin et~al., 2000]{LinLiuSloan:2000}
Lin, L., Liu, K., and Sloan, J. (2000).
\newblock A noisy {M}onte {C}arlo algorithm.
\newblock {\em Physical Review D}, 61.

\bibitem[Loh, 2003]{Loh2003}
Loh, W.-L. (2003).
\newblock On the asymptotic distribution of scrambled net quadrature.
\newblock {\em The Annals of Statistics}, 31:1282--1324.

\bibitem[Matousek, 1998]{Matousek:1998}
Matousek, J. (1998).
\newblock On the l2-discrepancy for anchored boxes.
\newblock {\em Journal of Complexity}, 14:527--556.

\bibitem[Niederreiter, 1992]{Niederreiter1992}
Niederreiter, H. (1992).
\newblock {\em Random Number Generation and Quasi-{Monte Carlo} Methods}.
\newblock Society for Industrial and Applied Mathematics, Philadelphia.

\bibitem[Nolan, 2007]{Nolan:2007}
Nolan, J. (2007).
\newblock {\em Stable Distributions: Models for Heavy-Tailed Data}.
\newblock Birkhauser, Boston.

\bibitem[Owen, 1997]{Owen:1997}
Owen, A.~B. (1997).
\newblock Scrambled net variance for integrals of smooth functions.
\newblock {\em The Annals of Statistics}, 25(4):1541--1562.

\bibitem[Pasarica and Gelman, 2010]{Pasarica:2003}
Pasarica, C. and Gelman, A. (2010).
\newblock Adaptively scaling the {M}etropolis algorithm using expected squared
  jumped distance.
\newblock {\em Statistica Sinica}, 20:343--364.

\bibitem[Peters et~al., 2012]{Peters:2012}
Peters, G., Sisson, S., and Fan, Y. (2012).
\newblock Likelihood-free {Bayesian} inference for $\alpha$-stable models.
\newblock {\em Computational Statistics \& Data Analysis}, 56(11):3743 -- 3756.

\bibitem[Pitt et~al., 2012]{Pitt:2012}
Pitt, M.~K., Silva, R.~S., Giordani, P., and Kohn, R. (2012).
\newblock On some properties of {Markov chain Monte Carlo} simulation methods
  based on the particle filter.
\newblock {\em Journal of Econometrics}, 171(2):134--151.

\bibitem[Quiroz et~al., 2016a]{Quiroz:2016}
Quiroz, M., Tran, M.-N., Villani, M., and Kohn, R. (2016a).
\newblock Exact subsampling {MCMC}.
\newblock Technical report.
\newblock https://arxiv.org/abs/1603.08232v2.

\bibitem[Quiroz et~al., 2016b]{quiroz:DelayedAcc}
Quiroz, M., Tran, M.-N., Villani, M., and Kohn, R. (2016b).
\newblock {Speeding up MCMC by delayed acceptance and data subsampling}.
\newblock {\em {Journal of Computational and Graphical Statistics}}.
\newblock accepted for publication.

\bibitem[Quiroz et~al., 2016c]{quiroz:villani:kohn::2016}
Quiroz, M., Villani, M., Kohn, R., and Tran, M.-N. (2016c).
\newblock Speeding up {MCMC} by efficient data subsampling.
\newblock Technical report.
\newblock http://arxiv.org/abs/1404.4178v3.

\bibitem[Shephard and Pitt, 1997]{Shephard:1997}
Shephard, N. and Pitt, M.~K. (1997).
\newblock Likelihood analysis of {non-Gaussian} measurement time series.
\newblock {\em Biometrika}, 84:653--667.

\bibitem[Sherlock et~al., 2015]{Sherlock:2015}
Sherlock, C., Thiery, A., Roberts, G., and Rosenthal, J. (2015).
\newblock On the efficiency of the pseudo marginal random walk {M}etropolis
  algorithm.
\newblock {\em The Annals of Statistics}, 43(1):238--275.

\bibitem[Stramer and Bognar, 2011]{Stramer:2011}
Stramer, O. and Bognar, M. (2011).
\newblock Bayesian inference for irreducible diffusion processes using the
  pseudo-marginal approach.
\newblock {\em Bayesian Analysis}, 6(2):231--258.

\bibitem[Tavare et~al., 1997]{Tavare:1997}
Tavare, S., Balding, D.~J., Griffiths, R.~C., and Donnelly, P. (1997).
\newblock Inferring coalescence times from {DNA} sequence data.
\newblock {\em Genetics}, 145(2):505--518.

\bibitem[Tran et~al., 2016]{Tran:2015}
Tran, M.-N., Nott, D.~J., and Kohn, R. (2016).
\newblock {Variational Bayes} with intractable likelihood.
\newblock {\em Journal of Computational and Graphical Statistics (accepted)}.

\bibitem[Vaart, 1998]{vandervaart1998}
Vaart, A.~W. (1998).
\newblock {\em Asymptotic statistics}.
\newblock Cambridge University Press, Cambridge (UK), New York (N.Y.).

\end{thebibliography}

\clearpage
\renewcommand{\theequation}{S\arabic{equation}}
\renewcommand{\thesection}{S\arabic{section}}
\renewcommand{\theproposition}{S\arabic{proposition}}
\renewcommand{\theassumption}{S\arabic{assumption}}
\renewcommand{\thelemma}{S\arabic{lemma}}
\renewcommand{\thecorollary}{S\arabic{corollary}}
\renewcommand{\thealgorithm}{S\arabic{algorithm}}
\renewcommand{\thefigure}{S\arabic{figure}}
\renewcommand{\thetable}{S\arabic{table}}
\renewcommand{\thepage}{S\arabic{page}}
\renewcommand{\thetable}{S\arabic{table}}
\renewcommand{\thepage}{S\arabic{page}}

\setcounter{page}{1}
\setcounter{section}{0}
\setcounter{equation}{0}
\setcounter{algorithm}{0}
\setcounter{table}{0}
\setcounter{lemma}{0}
\setcounter{figure}{0}

\section*{Online Supplement to \lq The Block Pseudo-Marginal Sampler\rq }
\appendix

\section{Proofs}\label{Proofs}
\begin{proof}[Proof of Lemma \ref{lemma: accept prob}]
We will show that
\begin{align} \label{eq: ratio}
\left ( \prod_{i=1}^G p_{\Usub{i}} (\usub{i}) \right )  q(\d \usub{1:G}^\prime|\usub{1:G})& =
\left (  \prod_{i=1}^G p_{\Usub{i}} (\usub{i}^\prime) \right) q(\d \usub{1:G}|\usub{1:G}^\prime)
\end{align}
Define the measure
\begin{align*}
\nu_j(\d \bm u_{(1:G)}, \d \bm u^\prime_{(1:G)}):= p_{U_{(j)}} (\bm u_{(j)})\d \bm u_{(j)} \prod_{k\neq j }\left (  \delta_{\bm u^\prime_{(k)} } (\d \bm u_{(k)} )
p_{U_{(k)}} (\bm u_{(k)})\d \bm u_{(k)}.
\right )
\end{align*}
It is straightforward to show that $\nu_j(\d \bm u_{(1:G)}, \d \bm u^\prime_{(1:G)})= \nu_j(\d \bm u^\prime_{(1:G)}, \d \bm u_{(1:G)})$ by showing that for any integrable
function $h(\bm u_{(1:G)}, \bm u^\prime_{(1:G)})$ with respect to $\nu_j(\d \bm u_{(1:G)}, \d \bm u^\prime_{(1:G)})$ we will have that
\begin{align*}
\int h(\bm u_{(1:G)}, \bm u^\prime_{(1:G)}) \nu_j(\d \bm u_{(1:G)}, \d \bm u^\prime_{(1:G)}) = \int  h(\bm u_{(1:G)}, \bm u^\prime_{(1:G)}) \nu_j(\d \bm u^\prime_{(1:G)}, \d \bm u_{(1:G)}).
\end{align*}
The result of the lemma now follows.
\end{proof}

It is useful to have the following definitions and results to obtain Lemmas~\ref{lemm: asymptotics} and \ref{lem: sigma opt}.
For any $\theta \in \Theta$, $\usub{i} \in \mathbb{U}_i$, $i=1, \dots, G$, we define
$\zsub{k}$, $z(\theta,\usub{1:G})$ and $z(\theta,\usub{i_1:i_k})$ as in Section~\ref{SS: pm based on z's}, and
$\zvec (\theta, \usub{1:G}):= \left (\zsub{1}(\theta,\usub{1}), \dots,\zsub{G}(\theta,\usub{G})\right )^{\rm T}$.

For $j=1, \dots, G$, let $\gsub{j}(\zsub{j}|\theta)$ be the density of $\zsub{j}$ when $\usub{j} $ has density $p_{\Usub{j} } (\cdot ) $,
and let  $g_Z(z|\theta)$ be the corresponding density of $z$.
The following lemma is a
straightforward generalization of the approach in \cite{Pitt:2012}.
\begin{lemma}\label{lemma: lemma on z's}
If the $\usub{j}$ are independent, each with density $p_{\Usub{j} } (\cdot ) $, for $j=1, \dots, G$, then
\begin{enumerate}
\item [(a)]
$\int \exp( \zsub{j} ) \gsub{j}(\zsub{j}|\theta)\d \zsub{j}  = 1 \quad \text{and} \quad \int \exp( z ) g_Z(z|\theta)\d z  = 1.$
\item [(b)]
$\ov \pi (\theta, \usub{1}, \dots, \usub{G} )= \prod_{j=1}^G \exp(z(\theta, \usub{j} ))p_{\Usub{j} } ( \usub{j})\pi(\theta),$
so that
\[\ov \pi (\usub{1:G} |\theta )= \prod_{j=1}^G \exp(z(\theta, \usub{j}))p_{\Usub{j} } ( \usub{j}).\]
Hence, conditional on $\theta$, the
$\usub{j}$ are independent in the posterior and have densities $\exp(z(\theta, \usub{j})) p_{ \Usub{j} } ( \usub{j})$.
\item [(c)]
$\ov \pi ( \zsub{1} , \dots, \zsub{G}|\theta)  =  \prod_{j=1}^G   \exp  ( \zsub{j}) \gsub{j}(\zsub{j}|\theta ) $
 so that, conditional on $\theta$,  $z_{(1)} , \dots, z_{(G)}$ are independent in the posterior with $z_{(j)}$ having density
$\exp  (\zsub{j}) \gsub{j}(\zsub{j}|\theta)$.
\item [(d)]
$\ov \pi(z|\theta)= \exp(z)g_Z(z|\theta)$ so that $\ov \pi(z,\theta) = \ov \pi(z|\theta)\pi(\theta) $.
 \end{enumerate}
\end{lemma}

\subparagraph{Pseudo-marginal MCMC based on $\zvec$}
Consider now the hypothetical pseudo-marginal MCMC sampling scheme on $\zvec $ with block proposal density for $\zvec$, conditional on $\theta$, given by
\begin{align} \label{eq: block proposal z}
q_Z(\zvec|\zvec^\prime, \theta):= \sum_{i=1}^G \omega_i g_i(\zsub{i}|\theta)\prod_{j\neq i} \delta_{\zsub{j}^\prime }(\d \zsub{j})
\end{align}
with $\omega_i = 1/G$. The proposal for $\theta$ is as above. Lemma~\ref{lemma: PM on vecz} below
shows that studying the optimality properties of the PM simulation based on $\theta$ and $\bm u $ is equivalent to studying it for
$\theta$ and $\zvec$. Although the PM based on the $\zvec$ is only \lq hypothetical\rq, as we usually cannot compute it,
we show below that it is more convenient to work with $\zvec$.
\begin{lemma} \label{lemma: PM on vecz}
\begin{enumerate}
\item [(a)] The acceptance probability \eqref{eq:acceptance 2} can be written as
\begin{align}\label{eq: accept prob u's}
\min\left \{  1,
\exp \left ( z(\theta, \usub{1:k-1}^\prime,\usub{k} , \usub{k+1:G}^\prime ) - z(\theta^\prime, \usub{1:G}^\prime) \right)
\frac{\pi(\theta) }{ \pi(\theta^\prime) }
\frac{q_\Theta(\theta^\prime|\theta)} {q_\Theta(\theta|\theta^\prime)}
\right \}
\end{align}
\item [(b)]
The acceptance probability of a PM scheme based on  $\zvec$ with proposal \eqref{eq: block proposal z} is equal to \eqref{eq: accept prob u's}.
Under Assumption~\ref{ass: perfect proposal}, it becomes
$
\min \left \{ 1, \exp(z- z^\prime) \right \} .
$
 \end{enumerate}
\end{lemma}
The following lemma and corollary are needed to prove Lemma~\ref{lemm: asymptotics}. Their proofs are straightforward and omitted.
\begin{lemma} \label{lemma: posterior}
Suppose that Assumption~\ref{ass: assumption on variances} holds. Then,
\begin{enumerate}
\item [(a)]
If the $ \usub{k}$ are independent and generated from $p_{\Usub{k}}(\cdot)$ for $k=1, \dots, G$, then
$z(\theta,\usub{1:G} )\sim \N(-\sigma^2/2,\sigma^2) $ and $z(\theta, \usub{1:k-1},\usub{k+1:G})\sim \N(-((G-1)/2G)\sigma^2 , ((G-1)/G)\sigma^2)$.
\item [(b)]
$\ov \pi (z(\theta,\usub{1:G}) |\theta)) =\N(z;\sigma^2/2,\sigma^2) $ and $\ov \pi ( z(\theta, \usub{1:k-1},\usub{k+1:G})|\theta) = \N(z; ((G-1)/2G)\sigma^2 , ((G-1)/G)\sigma^2)$.
\end{enumerate}
\end{lemma}

\begin{corollary}\label{corrol: corollar on implications}
Suppose that Assumptions~\ref{ass: assumption on variances} and \ref{assum: assumption on additional term} hold.
If $(\usub{1:k-1}^\prime,\usub{k+1:G}^\prime)\sim \ov \pi(\cdot |\theta)$
and $\usub{k}\sim p_{\Usub{k}}(\cdot) $ and they are independent,
then, $\zsub{k}(\theta, \usub{k}) + z(\theta, \usub{1:k-1}^\prime,\usub{k+1:G}^\prime) \sim \N(((G-2)/2G)\sigma^2, \sigma^2)$.
\end{corollary}

\begin{proof}[Proof of Lemma~\ref{lemm: asymptotics}]
The proof of the lemma
follows directly from Lemma~\ref{lemma: posterior} and Corollary~\ref{corrol: corollar on implications}.
\end{proof}

The next lemma gives the conditional and unconditional acceptance probabilities
of the \MH~scheme for $z$ and $\theta$.

\begin{lemma}\label{lem:acceptance prob}
Suppose Assumptions \ref{ass: assumption on variances} to \ref{ass: perfect proposal} hold and $\rho = 1-1/G$.
\begin{itemize}
\item [(i)]
The acceptance probability of the \MH{} scheme conditional on $z': = z(\theta',\usub{1:G}^\prime)$ is
\begin{align*}\label{eq:accept conditional}
\P({\rm accept}|z',\rho,\s) & =\exp(-x+\tau^2/2)\Phi\left(\frac{x}{\tau}-\tau\right)+\Phi\left(\frac{-x}{\tau}\right)
\end{align*}
with $x:=\bigg (z' + \frac{\sigma^2}{2}\bigg ) (1-\rho) $ and
$\tau:=\s\sqrt{ 1- \rho^2} $.
\item [(ii)] The unconditional acceptance probability of the \MH{} scheme is
\beq\label{eq:accept uncontional}
\P({\rm accept}|\rho,\s) = 2\left(1-\Phi \Big  (\frac{\s\sqrt{ 1 - \rho}}{\sqrt{2}}\Big )\right).
\eeq
\end{itemize}

\begin{proof}
We use the following results to obtain the conditional acceptance probability.
\begin{align}
\int_{-\infty}^A \exp(z) \N(z;a,b^2) dz & = \exp(a+b^2/2)\Phi\bigg ( \frac{A - a-b^2}{b}\bigg)  \label{eq: integ 1}\\
\int_A^\infty \N(z; a, b^2) dz & = \Phi\bigg ( \frac{a-A}{b}\bigg ), \label{eq: integ 2}\
\end{align}
where $\Phi(\cdot)$ denotes the standard normal CDF.
From Lemma~\ref{lemm: asymptotics}, we have that $a(z'):= \E(z|z') = -\sigma^2/2G + \rho z'$ and $\tau^2:= \Var(z|z') = \sigma^2(1-\rho^2)$, so that
the conditional
density of $z$ given $z'$ is $\N(z; a(z'), \tau^2)$. Using
\eqref{eq: integ 1} and \eqref{eq: integ 2}, the conditional probability of
acceptance is
\begin{align*}
\int \min(1, \exp(z-z')) & \N(z; a(z'), \tau^2)\d z  = \int_{-\infty}^ {z'} \exp(z-z')\N(z; a(z'), \tau^2) dz + \int_{z'} ^\infty \N(z; a(z'), \tau^2)\d z\\
& = \exp \bigg ( a(z') - z' + \tau^2/2 \bigg ) \Phi\bigg ( \frac{ z'-a(z') - \tau^2}{\tau} \bigg ) + \Phi\bigg ( \frac{a(z') - z'}{\tau} \bigg )\\
& = \exp \bigg ( -y + \tau^2/2 \bigg ) \Phi\bigg ( \frac{ y - \tau^2}{\tau} \bigg ) + \Phi\bigg ( \frac{-y}{\tau} \bigg ),
\end{align*}
where $y:= z' - a(z') = (1-\rho) (z'+ \sigma^2/2) $.

We now obtain the unconditional acceptance probability. We deduce from Lemma~\ref{lemm: asymptotics} that $z - z^\prime \sim \N ( - \sigma^2(1-\rho); 2\sigma^2(1-\rho) ) $.
The required result is now obtained using the identity $e^v \N(v; -a , 2a)= \N(v; a, 2a) $, with $a = \sigma^2 (1-\rho)$.
\end{proof}
\end{lemma}

We use the next lemma  to prove Lemma~\ref{lem: sigma opt}. It is of interest in its own right as it shows that under our assumptions
the inefficiency is independent of the function.
\begin{lemma}\label{lem: IF}
The inefficiency is given by
\begin{align} \label{eq: ineff}
\IF(\s,\rho)= 1 + 2 \E_{z' \sim \pi(z'|\sigma)}\bigg (  \frac{1-k(z'|\sigma, \rho)}{k(z'|\sigma, \rho) } \bigg ),
\end{align}
where $k(z'|\rho, \sigma)=\Pr({\rm accept}|z',\rho,\s)$ is the acceptance probability of the MCMC scheme conditional on the previous iterate $z'$ and is given by
part~(i) of Lemma~\ref{lem:acceptance prob}.

\begin{proof} 
For notational simplicity, we write the proposal density $q(z|z';\rho,\s)$ as $q(z|z')$,
 the acceptance probability $
\min \left \{ 1, \exp(z- z^\prime)\right  \} .
$ as $\alpha(z',z;\rho,\s)$
 as $\alpha(z',z)$ and the  acceptance
 probability $k(z'|\sigma, \rho)$,
 conditional on the previous iterate, as $k(z')$.
Let $\{(\t^{[j]},z^{[j]}),j=1, \dots, M\}$ be iterates, after convergence, of the Markov chain produced by the PM sampling scheme. Then,
the Markov transition distribution from $(\theta',z')$ to $(\theta,z)$ is
\bean
p(\t',z'; \d\t,\d z)&=&\a(z',z)\pi(\t)q(z|z')\d\theta \d z +\left(1-\int\a(z',z^*)\pi(\t^*)q(z^*|z')d\t^*dz^*\right)\delta_{(\t',z')}(\d \t,\d z)\\
&=&\a(z',z)\pi(\t)q(z|z')d\theta d z +\left(1-k(z'|\s,\rho)\right)\delta_{(\t',z')}(d\t,dz),
\eean
 $\delta_{\theta', z'}(\d\theta, \d z)$ is the
probability measure concentrated at $(\theta',z')$.

Consider now the space of functions
\begin{align*}
\mathfrak{F} = \bigg \{  \wt\varphi:\wt\Theta=\Theta\otimes\Bbb{R}\mapsto\Bbb{R}& ,\wt\varphi = \varphi(\theta)\psi(z), \pi(\varphi):=\E_{\theta \sim \pi(\theta)}(\varphi)  = 0,
\pi(\varphi^2):=\E_{\theta \sim \pi(\theta)}(\varphi^2) < \infty , \\
&   \pi(\psi^2):=E_{z \sim \pi(z)} (\psi)^2 < \infty\bigg \}.
  \end{align*}
We define the operator $P: \mathfrak{F}\mapsto \mathfrak{F}$ as
\begin{align*}
(P\wt\varphi)(\t,z)&: =\int\wt\varphi(\t^*,z^*)p(\t,z; \t^*,z^*)d\t^*dz^*\\
&=\pi(\varphi) \int \psi(z) \alpha(z,z^*)q(z^*|z) dz^* + \wt \varphi(\theta)(1-k(z))\\
& = \varphi(\theta)\psi(z)(1-k(z)).
\end{align*}
as $\pi(\varphi)=0$ by assumption. It is straightforward to check that $(P^j\wt\varphi)(\theta,z) = \varphi(\theta)\psi(z)(1-k(z))^j$ and
that $(P\wt \varphi)(\theta^{[j-1]}, z^{[j-1]}) = E(\wt \varphi(\theta^{[j]},z^{[j]})|\theta^{[j-1]},z^{[j-1]}) $. Hence,
$(P^j\varphi)(\theta_{0}, z_{0}) =\wt \varphi(\theta_0,z_0)(1-k(z_0))^j$.

We now consider $\wt \varphi(\theta, z) = \varphi(\theta)\psi(z) $ with $\psi(z) \equiv 1 $ so that $\wt \varphi \in \mathfrak{F}$;
suppose also that $(\theta_0, z_0)\sim \wt \pi_N$.
Define $c_j: = \Cov(\wt \varphi(\theta^{[j]},z^{[j]}), \wt \varphi(\theta^{[0]}, z^{[0]}) ) =
\Cov( \varphi(\theta^{[j]}),  \varphi(\theta^{[0]}) ) $. Then,
\begin{align*}
c_j & =  \E\big (\wt \varphi(\theta^{[0]},z^{[j]})\wt \varphi(\theta^{[0]}, z^{[0]})\big  )\\
& = \E_{(\theta^{[0]}, z^{[0]})\sim  \wt \pi_N} \big  (  \E (\wt \varphi(\theta^{[0]},z^{[j]})  |\theta^{[0]}, z^{[0]}  ) \wt \varphi(\theta^{[0]}, z^{[0]})\big ) \\
& = \E_{(\theta^{[0]}, z^{[0]})\sim  \wt \pi_N} \big ( (1-k(z_0))^j \wt \varphi( \theta^{[0]}, z^{[0]}  )^2\big ) \\
& = \E_{ z^{[0]} \sim  \wt \pi_N(z)} \big ( (1-k(z^{[0]}))^j)\E_{\theta^{[0]} \sim \pi} ( \varphi(\theta^{[0]})^2\big ) \\
\intertext{because $z^{[0]} $ only depends on $\sigma$ by construction}
& = \E_{ z^{[0]}\sim  \wt \pi_N(z)}\big ( (1-k(z^{[0]}))^j\big )c_0.
\end{align*}
The inefficiency $\IF$ is defined as
\begin{align*}
\IF & = (c_0 + 2 \sum_{j=1}^\infty c_j ) /c_0
 = 1 + 2\sum_{j=1}^\infty \E_{z \sim \wt \pi_N(z) } \bigg ( \big ( 1- k(z) \big )^j\bigg )
= 1 + 2\E_{z \sim \wt \pi_N(z) } \bigg ( \frac{1-k(z)}{k(z)} \bigg )
\end{align*}
as required.
\end{proof}
\end{lemma}

\begin{proof}[Proof of Lemma \ref{lem: sigma opt}]
From Lemma~\ref{lemm: asymptotics}, $\ov\pi(z'|\s)=\N(z';\s^2/2,\s^2)$.
Let $\omega: = [(1-\rho) (z'+\sigma^2/2) - \tau^2]/\tau$ with $\tau=\s\sqrt{1-\rho^2}$.
Then,
\begin{align*}
\omega & \sim \N\bigg ( -\frac{\rho \tau}{1+\rho} , \frac{1-\rho}{1+\rho} \bigg),
\end{align*}
and we note that the variance of $\omega$ just depends on $\rho$. For $\rho$ close to 1,
the variance of $\omega$ is approximately $1/(2G)$,  which is very small. Thus, $\omega$ will be concentrated close to its
mean $\omega^*: = -\rho \tau/(1+\rho)$.
Define
$p^*(\omega|\tau): = 1- k(z'|\rho, \sigma) = \Phi(\omega + \tau) + \exp(-\omega \tau - \tau^2/2) \Phi(\omega)$.
Then,
\begin{align*}
\IF(\sigma, \rho) & = \int \frac{1+p^*(\omega|\tau) } {1 - p^*(\omega|\tau)   } \N\bigg(\omega; -\frac{\rho \tau}{1+\rho} , \frac{1-\rho}{1+\rho}\bigg) d\omega.
\end{align*}
It is convenient to write $\IF(\sigma, \rho)$ as  $\IF(\tau|\rho)$, which we will optimize the computing time over $\tau$ keeping $\rho$ fixed.
Let,
\begin{align*}
f(\omega; \tau) & := \frac{1 + p^*(\omega|\tau) } {1 - p^*(\omega|\tau)}.
\end{align*}
Using the 4th order Taylor series expansion of $f(w;\tau)$ at
$\omega = \omega^*$, the inefficiency factor can be approximated by
\beqn
\IF_{\rm approx} (\tau|\rho) = f(\omega^*|\tau)+\frac12\frac{1-\rho}{1+\rho}f^{(2)}(\omega^*|\tau)+\frac18\left(\frac{1-\rho}{1+\rho}\right)^2f^{(4)}(\omega^*|\tau),
\eeqn
which is considered as a function of $\tau$ with $\rho$ fixed. This approximation is very accurate because, as noted,  the variance of $\omega$ is very small for large $G$.
So the computing time $\CT(\sigma,\rho)={\IF(\s,\rho)}/{\s^{1/\varpi}}$
is approximated by
\begin{align*}
\CT_{\rm approx} (\tau|\rho) & =(1-\rho^2)^{\frac{1}{2\varpi}}\frac{\IF_{\rm approx} (\tau|\rho)}{\tau^{1/\varpi}} \propto \frac{\IF_{\rm approx} (\tau|\rho)}{\tau^{1/\varpi}}
\end{align*}
Minimizing this term over $\tau$, for $\rho $ close to 1,
 we find that $\CT_{\rm approx} (\tau|\rho)$ is minimized at $\tau\approx 2.16$ for $\varpi=1/2$,
 and at $\tau\approx 0.82$ for $\varpi\approx3/2$.
 So the optimal $\s_\opt\approx 2.16/\sqrt{1-\rho^2}$ for $\varpi=1/2$ and $\s_\opt\approx 0.82/\sqrt{1-\rho^2}$ for $\varpi=3/2-\epsilon$ with any arbitrarily small $\epsilon$.

 For $\s_\opt\approx 2.16/\sqrt{1-\rho^2}$, the unconditional acceptance rate \eqref{eq:accept uncontional} is
\bean
P(\text{accept}|\rho,\s_\opt) &=& 2\left(1-\Phi \Big  (\frac{\s_\opt\sqrt{ 1 - \rho}}{\sqrt{2}}\Big )\right)\\
&=&2\left(1-\Phi \Big  (\frac{\s_\opt\sqrt{ 1 - \rho^2}}{\sqrt{2(1 + \rho)}}\Big )\right)\\
&\approx&2\left(1-\Phi \Big  (\frac{2.16}{{2}}\Big )\right)\approx 0.28.
\eean
Similarly, for $\s_\opt\approx 0.82/\sqrt{1-\rho^2}$, this probability is approximately 0.68.
\end{proof}

\section{Some large-sample properties of block PM for panel-data} \label{sec: large sample}
This section derives some properties of the block PM for large $T$ for the panel-data models discussed in Sections~\ref{SS: analysis of block}
and \ref{sub:PanelData} and shows that: (a)~the total number of samples required per MCMC iteration
is $O(T^{3/2})$ if MC is used;
and $O(T^{7/6})$  if RQMC is used, whereas the independent PM requires $O(T^2)$ samples;
 and (b)~we show that when $T$ is large the posterior correlation between $\theta$ and $z$ is weak. Since $\pi(\theta, z)$ is asymptotically (in $T$) multivariate  normal, this means that $\theta$ and $z$ are close to independent when $T$ is large, suggesting that moving $\bm u$ slowly, and hence moving $z$ slowly for a given $\theta$, does not greatly affect the mixing of the $\theta $ iterates.

Consider the panel-data model, with the panels in the $k$th  block
denoted by ${\G}_k$,  and suppose that we use the same $N_i=N_{k}$ samples for all panels $i \in {\G}_k$.
Let $L_i(\theta)=p(y_i| \theta)$ be the likelihood of the $i$th panel, and let $\wh L_{i}(\theta, \bm u_i)$ be the unbiased estimate of $L_i(\theta)$. We assume that
\begin{assumption} \label{ass: panel data assump}
For each $i \in {\G}_k$ and parameter value $\theta$, there exists an $A_i(\theta)^2$ such that as $ N_k \rightarrow \infty$,
\begin{align}\label{eq: clt panel}
N_{k}^{\varpi} \bigg ( \wh L_{i}(\theta, u_i) - L_i(\theta) \bigg )  & \stackrel{d}{\Rightarrow} \N(0, A_i(\theta)^2),
\end{align}
for some $\varpi>0$.
\end{assumption}
The central limit theorem \eqref{eq: clt panel} holds for most importance sampling estimates of the likelihood,
where $\varpi=1/2$ if MC is used and $\varpi=3/2-\eps$,  with an arbitrarily small $\eps>0$, if RQMC is used \citep[see, e.g.][]{Loh2003,Owen:1997}.

We now present a  result that supports the claim that for large $T$,
moving $\bm u$ slowly, and hence moving $z$ slowly given $\theta$, does not have an undesirable effect on the mixing of the $\theta$ iterates.
Its proof is in Appendix~\ref{Proofs}.
\begin{lemma}[Posterior orthogonality of $\t$ and $z$]\label{lemma:uncorrelated in n}
Suppose that the same number $N_T=O(T^{1/(4\varpi)})$ of samples is used for each panel and that
\begin{itemize}
\item[(i)] $\ov\pi(z|\t) = \N(z;-\zeta_T(\t)/2,\zeta_T(\t))$ with $\zeta_T(\t):=(T/N_T^{2\varpi})\eta_T(\t)$, $\eta_T(\t):=\frac{1}{T}\sum_{i=1}^T\gamma_i^2(\t)$.
\item[(ii)] $\E_\pi(\eta_T^2)<\infty$ and $\Var_\pi(\eta_T)=O(1/T)$.
\item[(iii)] $h(\t)$ is a function of $\t\in\Theta$ such that $\E_\pi(h^2)<\infty$, $\Var_\pi(h)=O(1/T)$ and $\Cov_{\pi}(h,\eta_T)=O(1/T)$.
\end{itemize}
Then, the posterior correlation of $h(\theta)$ and the log-likelihood estimation error $z$
is approximately zero for large $T$, i.e., $\Corr_{\ov\pi}(h,z)\to0$, as $T\to\infty$.
\end{lemma}
\noindent Assumption (i) in Lemma~\ref{lemma:uncorrelated in n} is justified by Lemma  \ref{lemm: clt log pi}, (ii)-(iii) are justified by the Bernstein von-Mises theorem~\citep[see][Section 10.2]{vandervaart1998}.

\begin{proof}[Proof of Lemma \ref{lemma:uncorrelated in n}]
Let $\mu_h:=\E_\pi(h)$ and $\mu_\eta:=\E_\pi(\eta_T)$. Then,
\begin{align*}
\Cov_{\ov\pi}(h,z)&=\E_{\ov\pi}[h(\theta) z]-\E_{\pi}[h(\theta)]\E_{\ov\pi}[z]\\
&=\frac{T}{2N_T^{2\varpi}}\left(\E_\pi[\eta_T(\t)h(\t)]-\E_\pi[\eta_T(\t)]\E_\pi[h(\t)]\right)\\
&=\frac{T}{2N_T^{2\varpi}}\Cov_\pi(h,\eta_T)\\
&=O\left(\frac{1}{N_T^{2\varpi}}\right),
\end{align*}
and,
\begin{align*}
\Var_{\ov\pi}(z)&=\frac14\Var_\pi(\zeta_T(\t))+\E_\pi (\zeta_T(\t))\\
&=\frac{T^2}{4N_T^{4\varpi}}\Var_\pi(\eta_T(\t))+\frac{T}{N_T^{2\varpi}}\mu_\eta\\
&=\frac{T}{N_T^{2\varpi}}\left(\mu_\eta+O\left(\frac{1}{N_T^{2\varpi}}\right)\right).
\end{align*}
\[\Var_{\ov\pi}(h)=\Var_{\pi}(h)=O(1/T).\]
Hence,
\begin{align*}
\Corr_{\ov\pi}(h,z)&=\frac{\Cov_{\ov\pi}(h,z)}{\sqrt{\Var_{\ov\pi}(h)\Var_{\ov\pi}(z)}}=O\left(\frac{1}{N_T^{\varpi}}\right)\to0
\end{align*}
as ${T\to\infty}$.
\end{proof}

\section{Derivation of the expression \eqref{eq: CT} for Computing Time}\label{app: CT derivation}
The average number of samples required in each MCMC iteration to give the same accuracy in terms of variance as $M$ iid iterates $\theta_1, \dots, \theta_M$
from $\pi(\theta)$ is proportional to
\begin{align*}
\frac1M \sum_{i=1}^M \sum_{k=1}^G N_k(\theta_i) \IF(\sigma, \rho) & = \frac1M \sum_{i=1}^M \sum_{k=1}^G \frac{G^{\frac{1}{2\varpi}} \gamma_{(k)}^{{1}/{\varpi}}(\theta_i) }{\sigma^{1/\varpi}}\IF(\sigma,\rho)
 \rightarrow\bigg ( G^{\frac{1}{2\varpi}} \sum_{k=1}^G \ov{\gamma^{1/\varpi}_{(k)}}  \bigg ) \frac{\IF(\sigma,\rho) }{ \sigma^{1/\varpi}}
\end{align*}
as $ M \rightarrow \infty$, where $\ov{\gamma^{1/\varpi}_{(k)}} = \E_{\theta \sim \pi} (\gamma_{(k)}^{1/\varpi}(\theta)) $.
The terms in the brackets are independent of $\s^2$, which means that  the computing time is
proportional to $\CT = \frac{\IF(\sigma,\rho) }{ \sigma^{1/\varpi}}$.

\section{An illustrative toy example} \label{S: A toy example}
This section uses a toy example to illustrate the ideas and results in Section~\ref{SS: analysis of block}.
Suppose that we wish to sample from $\ov\pi(\theta,z)=\pi(\theta)e^zg(z|\sigma)$
in which $\t$ is the parameter of interest, with $\pi(\t)=\N(\t;0,1)$ and $g_Z(z|\s)=\N(z;-\s^2/2,\s^2)$.
Suppose further that $z$ is divided into $G$ blocks so that $z=\sum_{k=1}^G z_{(k)}$ with $z_{(k)}\stackrel{iid}{\sim}\N(-\s_G^2/2,\s_G^2)$, $\s_G^2=\s^2/G$
and $G=100$.

We use the independent PM and the block PM
to sample from $\ov\pi(\theta,z)$ with $\sigma_G^2=2.34$, i.e. $\sigma^2=234$.
Suppose that $(\t',z')$ is the current state.
The proposal $(\t,z)$ in the independent PM is generated by $\theta \sim\pi(\theta)$ and $z\sim g(z|\sigma)$.
The proposal $(\t,z)$ in the block PM scheme is generated as follows.
Let $z'=\sum_{k=1}^Gz_{(k)}'$ be the current $z$-state
and let $k$ be an index uniformly generated from the set $\{1,...,G\}$.
Sample $\zsub{k}\sim \N(-\s^2_G/2,\s^2_G)$ and let $z=\sum_{j\not=k}z_{(k)}'+z_{(k)}$ be the proposal.
Both schemes accept $(\theta,z)$ with probability $\min(1,e^{z-z'})$.

\begin{figure*}[h]
\centerline{\includegraphics[width=.8\textwidth,height=.4\textwidth]{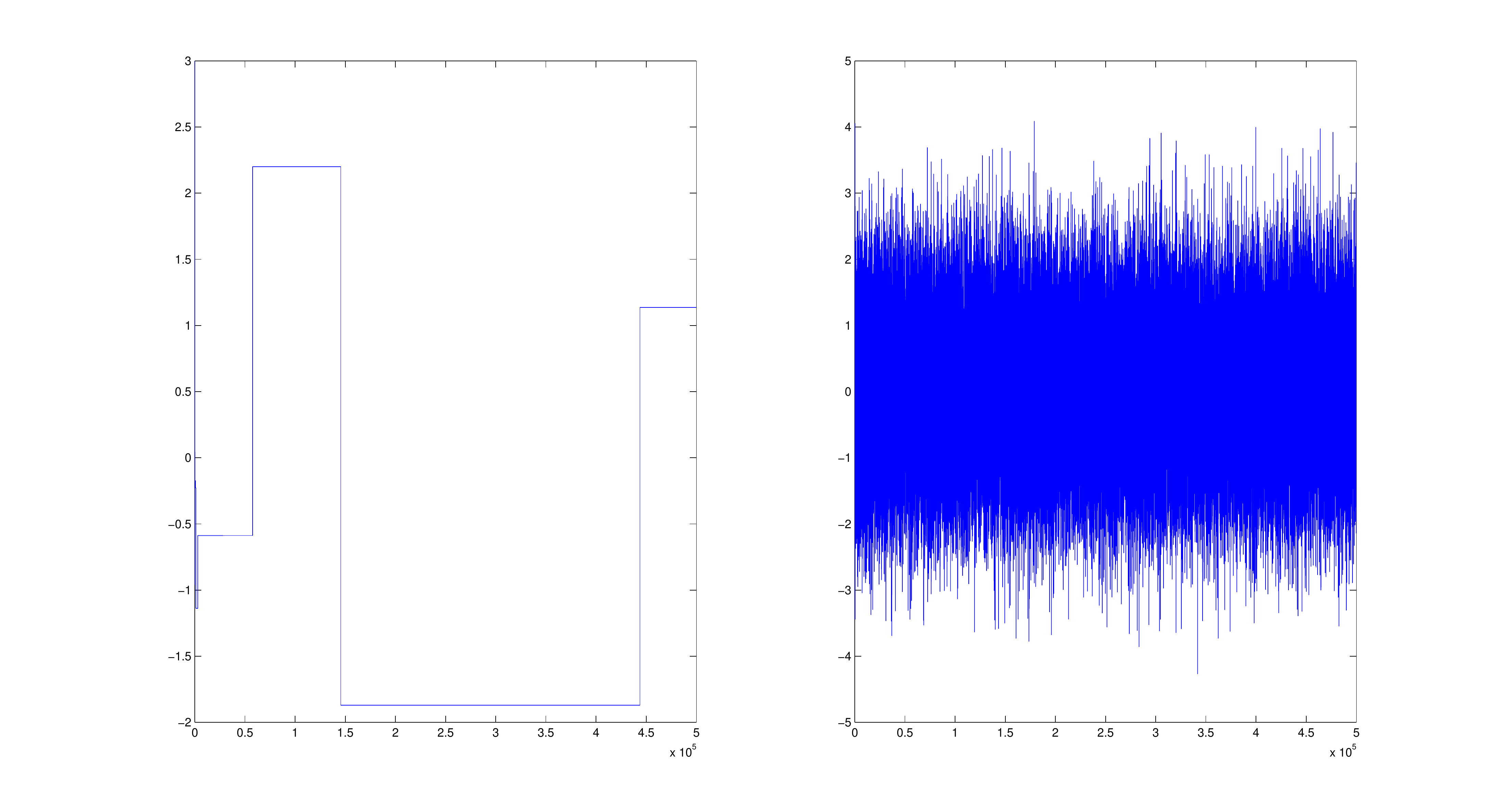}}
\caption{\label{toy example}
Toy example: The iterates of $\t$ generated by the independent PM scheme (left) and the block PM scheme (right). Both chains are initialised at 3 and run for 500,000 iterations.}
\end{figure*}
Figure \ref{toy example} plots the $\t$-samples generated by the independent PM scheme  and by the block PM scheme.
As expected, the independent PM chain is sticky because of the big variance  $\sigma^2=234$ of $z$.

We now study the effect of $\sigma^2$ on the acceptance rate and computing time $\CT(\s)$ of the block sampler.
Figure \ref{toy example 2} shows $\CT(\s)$ and the acceptance rates for various values of $\s^2$. The figure shows that
$\CT(\s)$ has a minimum value of 0.0263 at $\sigma^2=234$, where the acceptance rate is 0.279, which agrees with the theory.
\cite{Pitt:2012} show that the optimal value of $\s$ for the independent PM is around 1.
We also run this optimal independent PM scheme and obtain a value of the computing time $\CT(\s=1)=5.32$.
Hence, the optimal block PM is $5.32/0.0263 \approx 202$ times more efficient than the optimal independent PM.

\begin{figure*}[h]
\centerline{\includegraphics[width=.8\textwidth,height=.4\textwidth]{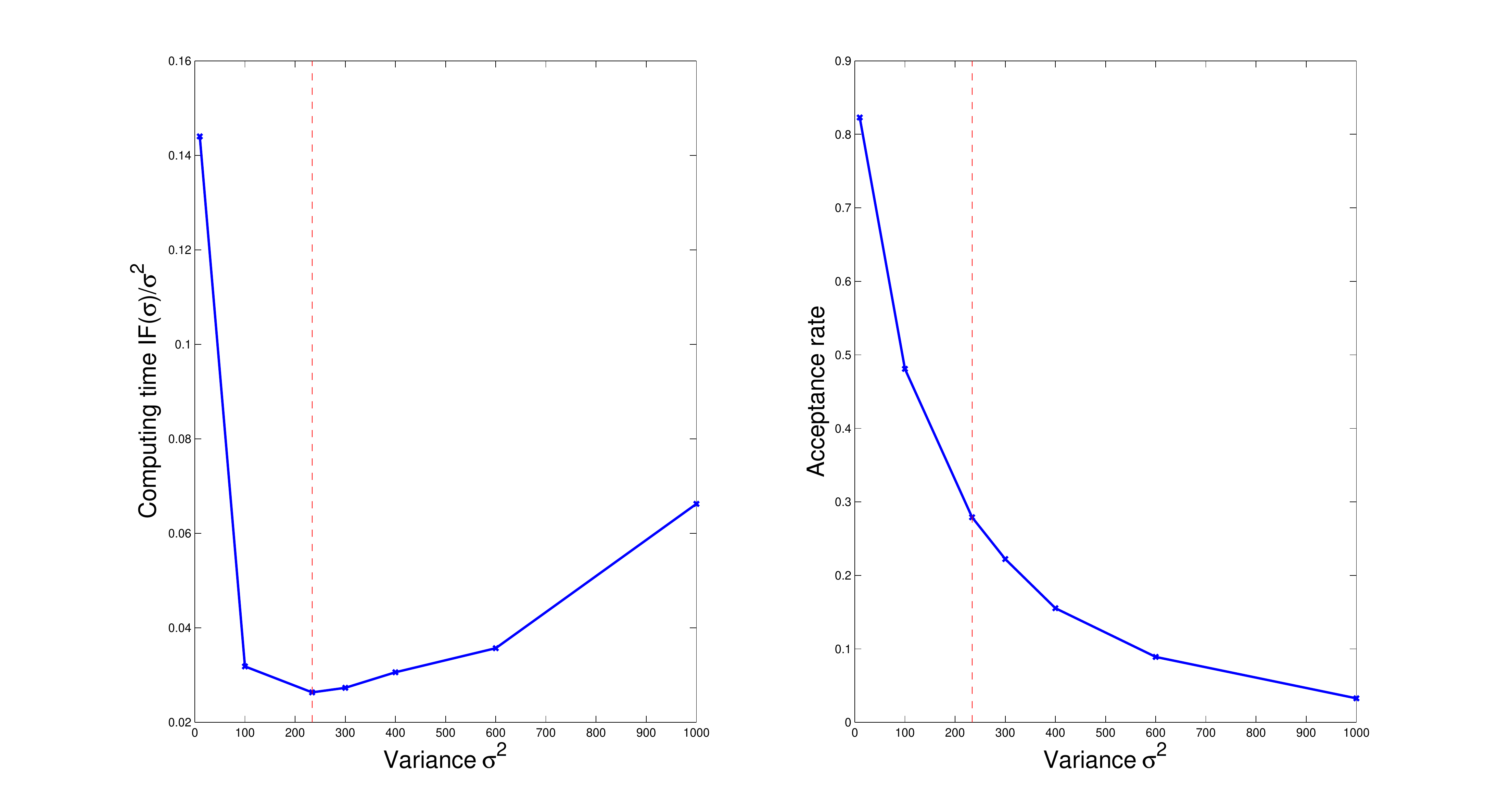}}
\caption{\label{toy example 2}
Toy example: The left panel shows the computing time $\CT(\s)$
and the right panel shows the acceptance rate v.s. the variance $\s^2$.
The dashed lines indicate the values w.r.t. the optimal variance $\s_{\opt}^2=234$.
}
\end{figure*}

\section{Further Applications} \label{S: further applications}
\subsection{ABC example}\label{sec:ABC example}
$\alpha$-stable distributions \citep{Nolan:2007} are heavy-tailed distributions
used in many statistical applications. The main difficulty when working with $\alpha$-stable distributions is that they do not have closed form densities,
which makes it difficult to do inference.
However, one can use ABC to carry out Bayesian inference \citep{Tavare:1997,Peters:2012},
because it is easy to sample from an $\alpha$-stable distribution.
Given the observed data $y$, ABC approximates the likelihood by its likelihood-free version
\beq\label{eq:likelihood free}
L_{\text{LF},\eps}(\t)=\int K_\epsilon(S(y'),S(y))p(y'|\t)\d y',
\eeq
where $K_\epsilon(\cdot,\cdot)$ is a kernel with the bandwidth $\epsilon$
and $S(\cdot)$ is a vector of summary statistics.
Inference is then based on the approximate posterior $p_\text{ABC}(\t|y)\propto p_\T(\t)L_{\text{LF},\eps}(\t)$,
where  $L_{\text{LF},\eps}(\theta)$ is unbiasedly estimated by  $\sum_{i=1}^M K_\epsilon(S(y^{[i]}),S(y))$,
 with $y^{[i]}\stackrel{iid}{\sim}p(\cdot|\theta)$.
Although the likelihood cannot be factorised as in \eqref{eq:panel data llh},
our example illustrates that the block PM scheme still applies.

We use the example in \cite{Peters:2012} and generate a data set $y=\{y_1,...,y_n\}$ with $n=200$ observations
from a univariate $\alpha$-stable distribution with parameters $\alpha=1.7$, $\beta=0.9$,
$\gamma=10$ and $\delta=10$.
The characteristic function $\phi_X(t)$ of a random variable $X$ following an $\alpha$-stable distribution with parameters $\alpha, \beta, \gamma$ and $ \delta$ is
\beq\label{eq:alphaStable}
\phi_X(t) = \begin{cases}
\exp \big(i\delta t- \gamma^\alpha |t|^\alpha \big[ 1 + i\beta \tan \frac{\pi\alpha}{2} \sgn(t)(|\gamma t|^{1-\alpha}-1)\big]\big) & \text{if } \alpha \neq 1 \\
\exp \big(i\delta t- \gamma |t| \big[ 1 + i\beta \frac{2}{\pi} \sgn(t)(\log(\gamma |t|) \big]\big) & \text{if } \alpha = 1. \\
\end{cases}
\eeq

We use the same summary statistics $S(y')=(\wh v_\alpha(y'),\wh v_\beta(y'),\wh v_\gamma(y'),\wh v_\delta(y'))$
 of a pseudo-dataset $y'=\{y_1',...,y_n'\}$
 as in
\cite{Peters:2012} and refer the reader to that paper for details.
We estimate $L_{\text{LF},\eps}(\theta)$ in \eqref{eq:likelihood free}
by $\wh L_{\text{LF},\eps}(\t)=K_\epsilon(S(y'),S(y))$, with $K_\epsilon$ the Gaussian kernel with covariance matrix $\epsilon I_4$,
using  only one pseudo-dataset ($M=1$) as  \cite{Bornn2016} show that $M=1$  is optimal.

Both the independent PM and the block PM were run for 50,000 iterations with the first 10,000 discarded as burn-in.
The block PM scheme was carried out as follows.
Given a vector of parameters $\theta$, write the pseudo-data point $y_i'$ as $f(\theta,u_i)$, with $u_i$ the set of MC random numbers used to generate $y_i'$.
We divide the set $u=\{u_i,i=1,...,200\}$ into $G=100$ blocks
with the $k$th block $u_{(k)}$ consisting of $u_{2k-1}$ and $u_{2k}$, $k=1,...,G$.

Table \ref{tab:ABC} summarises the performance measures for different values of $\epsilon$,
averaged over 10 runs.
In the table, the mean squared error (MSE) is the $l_2$-norm of the difference between
the estimated posterior mean and the true parameters.
See \cite{Pasarica:2003} for a definition of average squared jumping distance (ASD) as a performance measure
in MCMC. It is understood that the bigger the ASD the better.
The results show that the block PM performs better than the independent PM in this example.

\begin{table}[h]
\centering
\vskip2mm
{\small
\begin{tabular}{cccccc}
\hline 
$\epsilon$	&Methods	&Acc. rate 	&MSE	&IACT ratio	&ASD\\
\hline
10	&IPM		&0.31		&1.41	&2.07	&3.14\\
	&BPM		&0.37	     &1.29	&1 	&3.43\\
\hline
2 	&IPM 		&0.20		&1.14   &1.54	&0.70 \\
	&BPM		&0.30		&1.17	&1	&0.96\\
\hline
1 	&IPM 		&0.10		&0.96	&1.75	&0.18      \\
	&BPM		&0.21		&0.95	&1	&0.32\\
\hline
\end{tabular}
}
\caption{ABC example}\label{tab:ABC}
\end{table}
\subsection{State space example}\label{sec:time series example}

We consider a time series $\{y_t, t=1,...,T\}$ generated from the non-Gaussian state space model
\begin{eqnarray}\label{eq:poisson}
y_{t}|x_{t}&\sim&\text{Poisson}(\l_t),\;\;\l_t=e^{\beta+x_t},\\
x_{t+1}&=&\phi x_{t}+\eta_{t},\;\; \eta_{t}\sim N(0,\sigma^2),\;\; x_{1}\sim N(0,\sigma^2/(1-\phi^2)), \notag
\end{eqnarray}
with model parameters $\theta=(\beta,\phi,\sigma^2)$.
We generate the data using the parameter values $\beta=1$, $\phi=0.5$ and $\sigma^2=2(1-\phi^2)$, with $T = 1000$.

Following \cite{Shephard:1997} and \cite{Durbin:1997} we write the likelihood $(\theta)$ as
\begin{align*}
L(\theta) & = \int p(x_1|\theta)p(y_1|x_1, \theta) \prod_{t=2}^T p(x_t|x_{t-1}, \theta) p(y_t|x_t, \theta) \prod_{t=1}^T \d x_t
\end{align*}
We employ the high-dimensional importance sampling method of \cite{Shephard:1997} and \cite{Durbin:1997}
to obtain an unbiased likelihood estimator $\wh L(\t,u)$.
The simulation smoothing step requires $2T$ independent univariate normal variates to generate each sample path of the states,
so the set of random variates $\bm u$ needed is a matrix of size $N\times (2T)$, with $N$ the number of samples.
We divide $\bm u$ into $G=100$ blocks, where $\usub{1}$ consists of the first $2T/G$ columns of $\bm u$,
$\usub{2}$ consists of the next $2T/G$ columns of $\bm u$, etc.

We use the static strategy in this example, i.e. the number of sample paths $N$ is fixed.
Let $\bar\theta$ be some central value of $\theta$, e.g. the MLE estimate using the simulated maximum likelihood method \citep{Gourieroux:1995}.
For simplicity, we set $\bar\theta$ to the true value.
For the independent PM, we chose the value of $N$ so that $\Var(\wh L(\bar\theta,\bm u))\approx 1$,
where the variance $\Var(\wh L(\bar\theta,\bm u))$ is estimated by replication.
For the two block PM schemes,
one using MC and the using RQMC, we select $N$ such that $\Var(\wh L(\bar\theta,\bm u))\approx 2.16^2/(1-\rho^2)$
and $\approx 0.82^2/(1-\rho^2)$ respectively,
with the correlation $\rho$ estimated as follows.
Let $z=\log\;\wh L(\bar\t, \bm u)$ and $z'=\log\;\wh L(\bar\t, \bm u')$
with $\bm u'$ obtained from $\bm u$ by generating a new set for a randomly-selected block $\usub{(k)}$,
with the other blocks kept fixed. We generate $J=1000$ realisations $(z_{j},z_{j}')_{j=1}^J$ of $(z,z')$,
where a large value $N_0$ of $N$ is used,
and estimate $\rho$ by the sample correlation $\wh\rho$.
For the correlated PM of \cite{deligiannidis:doucet:pitt:2015}, we set the correlation $\varrho=0.99$ and use the same $N$ as in the block PM-MC.
Each MCMC scheme was run for 25,000 iterations including a burn-in of 5000 iterations.

\begin{table}[h]
\centering
\vskip2mm
{\small
\begin{tabular}{cccccc}
\hline 
Methods							&$N$		&Acceptance rate 	&IACT  ratio	&CPU ratio	&TNV ratio\\
\hline
IPM-MC 					&3500		&-				&-				&-				&-\\
CPM ($\varrho=.99$)	&56		&0.18			&1.613		&1.336		&2.155\\
BPM-MC		&56		&0.23			&1.154		&1.257		&1.451\\
BPM-RQMC			&16		&0.23			&1				&1				&1\\
\hline
\end{tabular}
}
\caption{State space example: performance measure ratios with the BPM-RQMC as the baseline} \label{tab:ss example}
\end{table}

Table \ref{tab:ss example} summarises the results.
We did not run the IPM-MC as it requires $N=3500$ samples,  which makes it too computationally demanding.
The block PM using RQMC performs the best.
Although both the correlated PM and block PM-MC use the same number of samples $N$,
the second requires less CPU time because it generates only one block of $\bm u$ in each iteration.

\end{document}